\documentclass[11pt]{article}

\usepackage[T5]{fontenc}

\usepackage{amsmath,amssymb,amsthm,algpseudocode,float,tkz-berge,caption,mathtools,thmtools,thm-restate,scalerel,stackengine,verbatim}
\usepackage{algorithm2e}
  
\stackMath
\newcommand\reallywidehat[1]{%
\savestack{\tmpbox}{\stretchto{%
  \scaleto{%
      \scalerel*[\widthof{\ensuremath{#1}}]{\kern-.6pt\bigwedge\kern-.6pt}%
          {\rule[-\textheight/2]{1ex}{\textheight}}
            }{\textheight}%
            }{0.5ex}}%
            \stackon[1pt]{#1}{\tmpbox}%
            }
\usepackage[margin=1in]{geometry}
\usepackage{multirow}
\usepackage{subcaption}
\usepackage{afterpage}
\usepackage{bm}
\usepackage[hidelinks]{hyperref}
\usepackage{tikz}
\usetikzlibrary{decorations.pathreplacing, positioning}

\usepackage{bbm}

\newcommand{\wt}[1]{\widetilde{#1}}
\newcommand{\Ot}{\widetilde{O}}
\newcommand{\Omegat}{\widetilde{\Omega}}

\newcommand{\abs}[1]{|#1|}
\newcommand{\eps}{\varepsilon}

\newcommand{\Z}{\mathbb{Z}}
\newcommand{\N}{\mathbb{N}}
\newcommand{\E}[2][]{\operatorname*{\mathbb{E}}_{#1 }\left\lbrack #2 \right\rbrack}
\newcommand{\Pb}[2][]{\operatorname*{\mathbb{P}}_{#1 }\left\lbrack #2 \right\rbrack}

\renewcommand{\NO}{{\bf NO~}}
\newcommand{\norm}[1]{\left\lVert#1\right\rVert}

\DeclareMathOperator{\var}{\text{Var}}

\DeclareMathOperator{\freq}{freq}
\DeclareMathOperator{\poly}{poly}
\DeclareMathOperator{\len}{len}
\DeclareMathOperator{\alg}{\mathcal{A}}

\DeclareMathOperator{\promise}{\mathtt{TrianglePromise}}
\newcommand{\Et}{E^\Delta}
\newcommand{\Vt}{V^\Delta}

\newtheorem{theorem}{Theorem}
\newtheorem*{theorem*}{Theorem}
\newtheorem{lemma}[theorem]{Lemma}
\newtheorem{definition}[theorem]{Definition}
\newtheorem*{definition*}{Definition}
\newtheorem*{lemma*}{Lemma}
\newtheorem{corollary}[theorem]{Corollary}
\newtheorem*{corollary*}{Corollary}

\newtheorem*{claim*}{Claim}

\title{Separations and Equivalences between Turnstile Streaming and
  Linear Sketching}
\author{John Kallaugher\thanks{This work was done in part while the authors were visiting the Simons Institute for the Theory of Computing.}\\UT Austin \and Eric Price\footnotemark[1]\\UT Austin}

{\makeatletter
	\gdef\xxxmark{%
		\expandafter\ifx\csname @mpargs\endcsname\relax 
		\expandafter\ifx\csname @captype\endcsname\relax 
		\marginpar{xxx}
		\else
		xxx 
		\fi
		\else
		xxx 
		\fi}
	\gdef\xxx{\@ifnextchar[\xxx@lab\xxx@nolab}
	\long\gdef\xxx@lab[#1]#2{{\bf [\xxxmark #2 ---{\sc #1}]}}
	\long\gdef\xxx@nolab#1{{\bf [\xxxmark #1]}}
}

\usepackage{etoolbox}

\makeatletter

\newcommand{\define}[4][ignore]{%
  \ifstrequal{#1}{ignore}{}{
  \@namedef{thmtitle@#2}{#1}}%
  \@namedef{thm@#2}{#4}%
  \@namedef{thmtypen@#2}{lemma}%
  \newtheorem{thmtype@#2}[theorem]{#3}%
  \newtheorem*{thmtypealt@#2}{#3~\ref{#2}}%
}

\newcommand{\state}[1]{%
  \@namedef{curthm}{#1}
  \@ifundefined{thmtitle@#1}{
  \begin{thmtype@#1}
    }{
  \begin{thmtype@#1}[\@nameuse{thmtitle@#1}]
  }
    \label{#1}
    \@nameuse{thm@#1}
  \end{thmtype@#1}
  \@ifundefined{thmdone@#1}{
  \@namedef{thmdone@#1}{stated}%
  }{}
}

\newcommand{\restate}[1]{%
  \@namedef{curthm}{#1}
  \@ifundefined{thmtitle@#1}{
    \begin{thmtypealt@#1}
    }{
  \begin{thmtypealt@#1}[\@nameuse{thmtitle@#1}]
  }
    \@nameuse{thm@#1}
  \end{thmtypealt@#1}
  \@ifundefined{thmdone@#1}{
  \@namedef{thmdone@#1}{stated}%
  }{}
}

\newcommand{\thmlabel}[1]{
  \@ifundefined{thmdone@\@nameuse{curthm}}{\label{#1}
    }{\tag*{\eqref{#1}}}
}
\makeatother

\begin{document}

\begin{titlepage}
  \maketitle
  \thispagestyle{empty}
  \begin{abstract}

    A longstanding observation, which was partially proven
    in~\cite{LNW14,AHLW16}, is that any turnstile streaming algorithm can
    be implemented as a linear sketch (the reverse is trivially true).
    We study the relationship between turnstile streaming and linear
    sketching algorithms in more detail, giving both new separations
    and new equivalences between the two models.

    It was shown in~\cite{LNW14} that, if a turnstile algorithm works
    for arbitrarily long streams with arbitrarily large coordinates at
    intermediate stages of the stream, then the turnstile algorithm is
    equivalent to a linear sketch.  We show separations of the
    opposite form: if either the stream length or the maximum value of
    the stream are substantially restricted, there exist problems
    where linear sketching is exponentially harder than turnstile
    streaming.

    A further limitation of the~\cite{LNW14} equivalence is that the
    turnstile sketching algorithm is neither explicit nor uniform, but
    requires an exponentially long advice string.  We show how to
    remove this limitation for deterministic streaming algorithms: we
    give an explicit small-space algorithm that takes the streaming
    algorithm and computes an equivalent module.
\end{abstract}
\end{titlepage}

\section{Introduction}
The study of streaming algorithms is concerned with the following question:
given a very large dataset that appears over time, what questions can one
answer about it without ever storing it in its entirety? Formally, one receives
$x \in \Z^n$ (e.g, the indicator vector for the set of edges in a graph) as a
series of updates $x_i \gets x_i + \Delta$ (e.g., edge insertions and
deletions).  One would like to estimate properties of the final vector $x$
while only ever using $o(n)$ space, ideally $\poly(\log n)$.  The space used
by the algorithm is the primary quantity of interest; other parameters such as
update or recovery time are often well-behaved as a matter of course for
small-space algorithms.  In this paper we focus on `turnstile' streams, where
$\Delta$ can be negative, as opposed to insertion-only streams, where it must
be positive.

The study of turnstile streaming has been very successful at revealing new
algorithmic techniques and insights.  It has found wide applicability,
with algorithms for a huge variety of problems.  Examples include norm
estimation in $\ell_2$~\cite{AMS96} or other
$\ell_p$~\cite{indyk2006stable,cormode2003comparing}; $\ell_0$
sampling~\cite{frahling2008sampling}; heavy
hitters~\cite{charikar2002finding,cormode2005improved}; coresets for
$k$-median~\cite{frahling2005coresets,indyk2011k}; and graph problems
such as finding spanning forests~\cite{AGM}, spectral
sparsifiers~\cite{KapralovLMMS14}, matchings~\cite{assadi2016maximum},
and triangle counting~\cite{TKMF09,PT12,KP17}.

Remarkably, for every single problem described above, the best known algorithm
is a \emph{linear sketch}, where the state of the algorithm at time $t$ is
given by a linear function of the updates seen to $x$ before time $t$.  And for
most of these problems, we know that the linear sketch is optimal.

Linear sketches have a number of other nice properties.  Their additivity means
that one can, for example, split a data stream across multiple routers and
sketch the pieces independently.  This has also made such sketches useful in
non-streaming applications such as distributed
computing~\cite{kapron2013dynamic}.  Their output depends only on the final
value of $x$, so they will work regardless of the length of the stream, the
order in which the stream arrives, and the intermediate states reached by the
stream. Their indifference to stream order means the randomness they use can
often be implemented with Nisan's
PRG~\cite{nisan1992pseudorandom,indyk2006stable}. 

They are also easier to prove lower bounds against, either using the
simultaneous message passing (SMP) model
(e.g.,~\cite{konrad2015maximum,assadi2016maximum,KKP18}) or additional
properties of linearity~\cite{price2012applications}.

So it would be nice if every turnstile streaming algorithm could be
implemented as a linear sketch.  And this \emph{is} true, as shown
in~\cite{LNW14}, but only subject to fairly strong limitations.  In
this paper, we explore the relationship in more detail.  First, we
show that some of the~\cite{LNW14} limitations are necessary: we
present natural problems with large, \emph{exponential} separations
between turnstile streaming and linear sketching with the limitations
removed.  Second, we show how to remove other~\cite{LNW14} limitations
for deterministic functions.

\paragraph{Separations between turnstile streaming and linear
  sketching.}  The result in~\cite{LNW14} requires that, in order for a
turnstile streaming algorithm to be equivalent to a linear sketch, the
streaming algorithm must be able to tolerate \emph{extremely} long
streams (longer than $2^{2^n}$) that reach correspondingly large
intermediate states.  In~\cite{AHLW16}, it was shown that this
equivalence can be extended to `strict' turnstile streams, where the
intermediate states never become negative but must still be allowed to
become extremely large in the positive direction.  However, the result
still leaves open the possibility of problems that require $\poly(n)$
space in linear sketching, but in turnstile streaming can be solved in
$O(\poly(\log n, \log \log L))$ space for length-$L$ streams, or in
$O(\poly(\log n, \log M))$ space for streams whose intermediate state
never leave $[-M, M]^n$ (a `box constraint').

Such a box constraint is particularly natural in graph streaming: if
the stream represents insertions and deletions of edges in a graph,
then the intermediate states $x$ should lie in
$\{0, 1\}^{\binom{n}{2}}$.  At the same time, graph streaming is where
a theorem on equivalence between streaming and sketching would be most
useful: most of the problems for which we have lower bounds on linear
sketches but not turnstile streaming involve graphs.  The~\cite{LNW14}
equivalence gives lower bounds for these problems, but only for
turnstile algorithms that are indifferent to stream length and
tolerate multigraphs at intermediate stages.

The conjunction here, where the box constraint is most relevant in
precisely the situations where we have no alternative lower bounds
to~\cite{LNW14}, suggests an opportunity: perhaps we have not found
direct turnstile streaming lower bounds for these problems because no
such lower bounds exist that respect the natural constraints of
graphs.  Maybe better algorithms exist, and we just haven't found them
because they require substantially different, nonlinear approaches to
turnstile sketching.

In this paper, we show that this can in fact be the case, presenting
natural assumptions on adversarially ordered turnstile streams for
which we can prove exponential separations between turnstile streaming
and linear sketching.  We give several different settings in which
there are problems that can be solved with an $O(\log n)$ space
streaming algorithm, but for which any linear sketch requires
$\Omegat(n^{1/3})$ space.

We first consider \emph{binary} streams: the data stream can have
arbitrary length, but must lie in $\{0, 1\}^n$ at all times.  We
present a problem that can be solved over such streams in $O(\log n)$
space, but requires $\Omega(n^{1/3} / \log^{1/3} n)$ space to solve in
linear sketching.

We then consider \emph{short} streams: the data stream can have arbitrary
intermediate states, but only $L = O(n)$ updates.  We show that for a natural
problem---triangle counting on bounded degree graphs with many triangles---an
$O(\log n)$ space streaming algorithm is possible, while any linear sketching
algorithm takes $\Omega(n^{1/3})$ space.  The streaming algorithm depends
polynomially on $L$, and a separation remains for any $L = o(n^{7/6})$.

The only previously known separation between turnstile streaming and linear
sketching is due to Jayaram and Woodruff~\cite{jayaram2018data}, which for
$\ell_1$ estimation with $L = O(n)$ gives a separation of $O(\log \log n)$ vs
$\Theta(\log n)$.  While that is also an exponential separation, it would be
consistent with, for instance, turnstile algorithms being convertible to linear
sketches with an additive $O(\log n)$ loss.

Section~\ref{sec:separationresults} describes these results more
formally, as well as two other similar results.

\paragraph{An explicit, computable reduction for deterministic algorithms.}
Another limitation of~\cite{LNW14}, as well as the earlier
work~\cite{ganguly2008lower} that applies to deterministic streaming
algorithms, is that the reduction is not explicit.  These reductions
show the \emph{existence} of a linear sketch, and corresponding
recovery algorithm, that are equivalent to the streaming algorithm;
they do not show that the sketch and recovery algorithm can be
computed at all, much less computed in small space.  The distinction
is analogous to that of $L/\text{poly}$ and $L$: they are linear
sketching algorithms with a very long advice string.  For an $s$-bit
linear sketching algorithm, the advice string needs $ns$ bits for the
sketch and $2^s$ bits for the recovery.  This is typically referred to
as a ``nonuniform'' result, but note that the advice string is much
longer than the algorithm is supposed to store: there does not
necessarily exist an $O(s)$-bit machine that computes the linear
sketch for each input size $n$ and space-$s$ streaming algorithm.

We show for deterministic streaming algorithms how to perform an
\emph{explicit} reduction: given an $s$-bit streaming algorithm, our
algorithm computes an equivalent $s$-bit linear sketching algorithm in
$O(s \log n)$ bits of space.  To do so, we generalize what a ``linear
sketch'' means from prior work:

\begin{definition}\label{def:linearsketch}
  A \emph{linear sketch} consists of a $\Z$-module homomorphism $\phi$ from
  $\Z^n$ to a module $M$.
\end{definition}

The ``standard'' linear sketch is $\phi(x) = Ax \bmod q$ for some
matrix $A \in \Z^{m \times n}$ and set of moduli $q \in \Z_+^m$; the
corresponding module $M$ is
$\Z_{q_1} \times \dotsb \times \Z_{q_m}$\footnote{Some descriptions of
  linear sketches, such as the introduction of~\cite{LNW14}, omit the
  moduli $q_i$.  But then the sketch would not have bounded space on
  all streams, so these works end up introducing moduli either
  explicitly (as in~\cite{LNW14}) or implicitly (by storing
  coordinates as $O(\log n)$-bit words with overflow).  Other authors,
  such as~\cite{AHLW16}, include the moduli.}.  But
Definition~\ref{def:linearsketch} captures the ways in which standard
linear sketching is useful: the sketch is linear
($\phi(x + y) = \phi(x) + \phi(y))$), and therefore mergeable and
indifferent to stream length and order.

In fact, according to the structure theorem for finitely generated
$\Z$-modules, every linear sketch to a finite module $M$ is equivalent
to a standard linear sketch with $q_1 | q_2 | \dotsb | q_m$ using the
same space (i.e., $\log \abs{M} = \log \prod q_i$).  However, we do
not know how to compute this transformation efficiently, and our
algorithm creates a linear sketch with $\phi$ and $M$ of a different
form.

\define{thm:detreductionpartial}{Theorem}{%
  Suppose there is a deterministic algorithm solving a streaming
  problem $P$ that works on streams of all lengths, uses $S$ space
  during updates and recovery, and uses $s$ space between
  updates. Then there is a linear sketching algorithm for $P$ that
  uses $O(S + s \log n)$ space during updates and recovery, and stores an
  $s$ space sketch.
}
\state{thm:detreductionpartial}

This reduction still has the stream length and box constraint
limitations discussed in the previous section, but they are actually
somewhat weaker than~\cite{LNW14,AHLW16}---the length required is
exponential in $s$, not doubly exponential.  As with these works, the
above theorem applies to streaming problems representing general
binary relations: any given input may have multiple valid outputs (as
in approximation algorithms) or even consider every output to be valid
(as in promise problems, where some inputs are invalid).  For the more
restrictive setting of \emph{total functions}, where every input has a
single valid output, we can remove the restriction on stream length:
the same result holds for algorithms that work on streams of length
$n + O(s)$.

Another advantage we believe our reduction has over prior ones is
that, because it is explicit, it is easier to understand---and to
understand the limitations of.  We hope that this makes it easier to
develop new turnstile algorithms that circumvent the limitations of
these lower bounds.

We now present the definitions required to state our results more
formally.

\subsection{Definitions}\label{sec:definitions}

\begin{definition}
\label{dfn:streamproblem}
  A \emph{data stream problem} is defined by a relation
  $P_n \subseteq \Z^n \times \Z$.  A \emph{turnstile data stream}
  $\sigma$ of length $L$ is a sequence of updates
  $\sigma_1, \dotsc, \sigma_L \in [n] \times \Z$.  The
  \emph{state} of a stream at time $t$ is given by
  \[
    x^{(t)} := \freq \sigma^{(t)} := \sum_{(i, \Delta) \in \{\sigma_1, \dotsc, \sigma_t\}} \Delta \cdot e_i.
  \]
  and the final state is $x = \freq \sigma^{(L)}$.
\end{definition}
We will also write $\len(\sigma)$ for $L$.

\begin{definition}
  A \emph{data stream algorithm} $\alg$ is defined by a random
  distribution on initial states $y$; a transition function that takes
  a state $y$ and a stream update $\sigma_i$ and returns a new state
  $y'$; and a (possibly randomized) post-processing function $g$ that
  takes the final state $\alg(\sigma)$ and returns an output
  $g(\alg(\sigma))$.

  We say that $\alg$ \emph{solves} a problem $P_n$ under condition $C$
  if, for all streams $\sigma \in C$, with 2/3 probability,
  $(\freq \sigma, g(y)) \in P_n$.  We say that $\alg$ uses $s$ space
  between updates if all states reached by $\alg$ while processing
  $\sigma$ can be represented in $S$ bits of space; we say it uses
  $S \geq s$ space during updates and recovery if the transition
  function and post-processing function use $S$ space.
\end{definition}

One very common stream condition considered in the literature is that
of `strict' turnstile streams, where $x^{(t)}_i \geq 0$ for all $i$
and $t$.  The goal of our separations is to describe relatively mild
stream conditions under which turnstile streaming is much easier than
linear sketching.  The goal of our equivalences is to bound $S$ as
well as $s$ in the reduction.

For the explicit problems we consider, which are decision and counting
problems, the set of valid outputs for each input forms an interval.
Therefore the success probability can always be amplified to
$1-\delta$ by taking the median of $O(\log 1/\delta)$ repetitions.

\begin{definition}
  A \emph{linear sketching algorithm} is a data stream algorithm where
  the state is $\phi(\freq \sigma)$, where $\phi$ is a linear sketch,
  along with the randomness used to choose $\phi$.
\end{definition}

We will at times need some ``standard'' streams constructed from
vectors or from other streams:
\begin{definition}
For any $x \in \Z^n$, the ``canonical'' stream $\kappa(x)$ is the stream that
inserts each of its coordinates in order, skipping any zero coordinates, so
$\len(\kappa(x)) = ||x||_0$.

For any stream $\sigma$, $\overline{\sigma}$ is the stream with the
same sequence of updates but the opposite sign on each update, so if
$\sigma_t = (i, \Delta)$, $\overline{\sigma}_t = (i, -\Delta)$.
\end{definition}

For certain reductions we will need to iterate through (subsets of) $\N^n$ in
``little-endian'' order, that is, $x < y$ if $x_n < y_n$, or $x_n = y_n$ and
$x_{n - 1} < y_{n-1}$, and so on.

\subsection{Our Results: Separations}\label{sec:separationresults}

\paragraph{Box-constrained streams.}  Our first result concerns binary streams,
in which we are promised that the partial stream states $x^{(t)}$ lie in $\{0,
1\}^n$ at all times.

\begin{definition}[Box constraint]
  $\Gamma_M$ is the set of streams such that for all times $t$,
  $\norm{x^{(t)}}_\infty \le M$. $\Gamma_{0,1}$ is the set of streams
  such that for all times $t$, $x^{(t)} \in \lbrace 0, 1\rbrace^n$.
\end{definition}

\define{thm:01-clean}{Theorem}{%
For every $n \in \mathbb{N}$, there exists a data stream problem $P_n
\subseteq \lbrace 0, 1\rbrace^n \times \lbrace 0,
1\rbrace$ such that:
\begin{enumerate}
\item Any linear sketching algorithm solving $P_n$ requires $\Omega(n^{1/3} /
\log^{1/3} n)$ bits of space.
\item There exists a turnstile streaming algorithm that solves $P_n$
  on $\Gamma_{0,1}$ in $O(\log n)$ space.
\end{enumerate}
}%
\state{thm:01-clean}%

Note that as the final state of a linear sketching algorithm depends only on
the final state of the stream, any linear sketching algorithm solving $P_n$
on $\Gamma_{0,1}$ would also solve $P_n$ for arbitrary streams.

One property of binary streams is that every update to a coordinate $i$
uniquely identifies the value of $x_i$ after the update.  Over larger domains,
this is no longer true.  We can still show a similar result for inputs of size
$m$, as long as intermediate results never exceed $2M-1$:

\define{thm:2m-clean}{Theorem}{%
For every $M, n \in \mathbb{N}$, there exists a data stream problem $P_n
\subseteq \lbrace -M, \dotsc, M\rbrace^n \times \lbrace 0,
1\rbrace$ such that:
\begin{enumerate}
\item Any linear sketching algorithm solving $P_n$ requires $\Omega(n^{1/3} /
\log^{1/3} n)$ bits of space.
\item There exists a turnstile streaming algorithm that solves $P_n$
  on $\Gamma_{2M-1}$ in $O(\log n \log M)$ space.
\end{enumerate}
}%
\state{thm:2m-clean}%

Interestingly, this $2M$ threshold matches one of the results
in~\cite{AHLW16}.  Recall that one requirement for~\cite{LNW14} to show an
equivalence between linear sketching and streaming is that the
streaming algorithm tolerate intermediate states of (more than) doubly exponential
size, i.e., $\Gamma_{2^{2^{n}}}$.  One result in~\cite{AHLW16} shows
that this can be relaxed to $\Gamma_{2M}$---as long as $M > 2^{ns}$,
where $s$ is the algorithm space.  That additional requirement is very
strong (e.g., one cannot store a single coordinate of the input) but
if it did not exist, the result would imply that our $2M-1$ threshold
cannot be increased.

\paragraph{Graph streams}
Our separations for binary and box-constrained streams are based on a somewhat
unnatural problem.  We also present separations for a more natural problem,
that of counting triangles in bounded-degree graphs.

In this problem, the final state $x \in \{0, 1\}^{\binom{n}{2}}$ represents a
graph of maximum degree $d$.  In the \emph{counting} version of the problem,
one would like to estimate the number of triangles $T$ in the graph to within a
multiplicative $1 \pm \eps$ factor with probability $2/3$; in the
\emph{decision} version, one would like to determine whether the number of
triangles is zero or at least $T$.

In the insertion-only model of computation, the counting problem can
be solved in $O(d \frac{m}{\eps^2 T} \log n)$ space~\cite{PTTW13},
where $m \le nd/2$ is the number of edges in the graph, while in the
linear sketching model it requires $\Omega(n/T^{2/3})$ space even for
the decision version with $d = 2$~\cite{KKP18}.  This leaves a natural
question: for constant $d$ and linear $T$, do turnstile streaming
algorithms require $\log n$ or $n^{1/3}$ space?  We show, under
natural conditions on the stream, that it is the former.

In our first result on this problem, we suppose that the stream represents a
bounded degree graph at all times, not just at the end of the stream.  In this
model, we can match the best known complexity in the insertion-only model for
\emph{constant-degree} graphs~\cite{JG05, PTTW13}.

\define{thm:maxd}{Theorem}{%
  There is a streaming algorithm for triangle counting in max-degree
  $d$ graphs, over streams with intermediate states of max degree $d$,
  that uses $O\left(\frac{d^2m}{\eps^2T}\log n\right)$ bits.
}
\state{thm:maxd}

When $T$ is $\Theta(n)$, this is $O(d^3 \log n)$: exponentially
smaller than the $\Omega(n^{1/3})$ lower bound for linear sketching
for constant degree graphs, and still separable up to small polynomial
degrees.

In our second result on this problem, we suppose that the total length
of the stream is $L$, but allow the intermediate states to be
arbitrary multigraphs.

\define{thm:boundedl}{Theorem}{%
  There is a streaming algorithm for triangle counting in max-degree
  $d$ graphs of length-$L$ streams using
  $O\left(\frac{d^3L^2}{\eps^2T^2}\log n\right)$ bits of space.
}
\state{thm:boundedl}

For constant degree graphs with $L$ and $T$ both $\Theta(n)$, this is again
$O(\log n)$ rather than the $\Omega(n^{1/3})$ required by linear sketching.
Note that in the graph setting, $n$ is the number of vertices (equivalently
edges, as the degree is constant), and so $L = O(n)$ is equivalent to saying
that at least a constant fraction of the insertions in the stream are never
followed by a corresponding deletion; this is a reasonable assumption for real
world graph streams such as the Facebook friends graph.

\subsection{Our Results: Equivalences}\label{sec:equivalenceresults}

Our main equivalence result is Theorem~\ref{thm:detreductionpartial}.
We also have a slightly stronger theorem for total functions:

\define{thm:detreductiontotal}{Theorem}{%
  Suppose there is a deterministic algorithm solving a streaming problem $P$
   that works on streams of length $n + 2s + 2$, uses $S$ space during updates
   and recovery, and uses $s$ space between updates. If $P$ corresponds to a
   \emph{total function} on $\Z^n$, there is a linear sketching algorithm for
   $P$ that uses $O(S + s \log n)$ space during updates and recovery, and
   stores an $s$ space sketch.
 }
\state{thm:detreductiontotal}

The advantage of this over Theorem~\ref{thm:detreductionpartial} is
that the stream length is short (i.e., $(1 + o(1))n$).  The downside
is that total functions are much more restrictive than binary
relations, excluding approximation and promise problems.

Relative to~\cite{LNW14,AHLW16}, the main benefit of
Theorem~\ref{thm:detreductionpartial} is that the reduction is
explicitly computable in small space.  The downside is that it only
applies to deterministic streaming algorithms, not randomized ones.
But note that even those reductions are limited in the extent to which
they apply to randomized algorithms: they assume that the randomness
is stored in the $s$ space used by the algorithm.  As a result, they
do not apply to algorithms that flip a coin on every update, or even
ones that sample a random update from the data stream: such algorithms
use $L$ and $\log L$ bits of randomness, respectively, which are much more
than $s$ for the streams considered in the reduction.

\section{Related Work}

\paragraph{Equivalences between streaming and linear sketching.} As described
above,~\cite{LNW14}, building on~\cite{ganguly2008lower}, proved that any
turnstile streaming algorithm can be implemented as a linear sketch, assuming
the streaming algorithm can tolerate arbitrarily long streams that feature
arbitrarily complicated intermediate states.  The followup work~\cite{AHLW16}
removed or relaxed some of the restrictions on this equivalence: for example,
they show that it still holds if the algorithm only
works in the `strict' turnstile model where all intermediate states are
non-negative. They also show that it holds if the algorithm only tolerates
exponentially large (in the space usage of the algorithm and the dimension of
the problem) intermediate values, rather than doubly exponentially large ones.

Another line of work on the problem has considered XOR streams or
other modular updates~\cite{kannan2016linear,hosseini2018optimality}.
XOR streams are like binary streams, except that insert and delete
updates are indistinguishable.  For such
streams,~\cite{hosseini2018optimality} shows that for total functions
the equivalence between streaming and linear sketching holds under
much more mild assumptions: as long as the algorithm works on streams
of length $\Ot(n^2)$.  As with all the other existing equivalences,
these are nonuniform: they do not show that the linear sketching
algorithm is efficiently computable\footnote{There is some discussion
  in~\cite{hosseini2018optimality} of generating the linear functions
  in small space, but this only refers to the space used to store the
  randomness; even in the deterministic setting, the construction is
  nonuniform.
}.

\paragraph{Lower bounds for linear sketches.}
The most common lower bound technique in streaming algorithms is the
construction of reductions to one-way communication complexity. One encodes a
hard one-way communication complexity problem into a stream by encoding Alice's
input into the first half of the stream, and Bob's input into the second half.
If a solution to the streaming problem yields a solution to the communication
problem, this yields a lower bound on the streaming algorithm's space.  The
hard instances created by this approach tend to be fairly nice: the stream
length is never more than $2n$, for example.

For linear sketching, lower bounds may also be proved by reductions to the more
restrictive simultaneous message passing (SMP) model.  Rather than Alice
sending a short message to Bob, Alice and Bob must both send a short message to
a referee, who adds their sketches to solve the problem.  (One may also have
more than two parties, which is typically more fruitful in the SMP model than
in the one-way communication model.)

These lower bounds translate into turnstile streaming lower bounds
using~\cite{LNW14,AHLW16}, but the instances become horrible, leading
to weak implications. In particular, this approach can never rule out
algorithms using either $O(\log \log L)$ or $O(\log M)$ space, for
length-$L$ streams with intermediate states that never leave the
$[-M, M]^n$ box.

Still, for a number of problems we only know how to get strong lower
bounds via linear sketching.  Examples include finding approximate
maximum matchings~\cite{konrad2015maximum,assadi2016maximum},
estimating the size of the maximum
matching~\cite{assadi2017estimating}, subgraph counting~\cite{KKP18},
and finding spanning forests~\cite{nelson2019optimal}.  Most such
problems are graph problems, but the translation of the lower bound
from linear sketching to streaming only applies if intermediate states
are allowed to be multigraphs.

\paragraph{Non-linear turnstile algorithms.}
We are aware of one case of a turnstile streaming algorithm that is not
implementable in linear sketching.  

Jayaram and Woodruff~\cite{jayaram2018data} consider problems on
data streams with a bounded ratio of deletions to insertions (this is similar
to our condition in Theorem~\ref{thm:boundedl}, as a long stream requires a
large ratio of deletions to insertions and vice versa).  The precise result
depends on the problem, but roughly speaking: if the final magnitude of the
vector is at least $\alpha < 1/2$ times the sum of the magnitudes of all the
updates, the space complexity can be improved over linear sketches by a factor
of $\log_{\alpha} n$.  In particular, for $\ell_1$ estimation, an exponential
separation can be obtained, but this is $O(\log n)$ vs. $O(\log\log n)$, so
even the harder case requires very little space. 

Furthermore, these results do not rule out~\cite{LNW14} being extended to short
streams, as~\cite{LNW14} requires the algorithm to store all the random bits it
ever uses (in contrast to the normal setting where only random bits that are to
be reused have to be stored). The algorithms in~\cite{jayaram2018data} use 
(non-reused) randomness to sample from the updates they see, and so under this 
constraint they would end up needing substantially larger space. By contrast,
our algorithms use only a small amount of randomness relative to their space,
so they do show that a length constraint is necessary for~\cite{LNW14}.

\section{Overview of Techniques}

\subsection{Turnstile-Sketching Separations}
\subsubsection{Binary and Box-Constrained Streams}

\paragraph{Binary streams.}  To prove Theorem~\ref{thm:01-clean}, we embed a
hard communication problem from~\cite{KKP18} into a binary stream.  In this
communication problem, which we call $\promise(n)$ and illustrate in
Figure~\ref{fig:instance2}, there are three players and $O(n)$ vertices, each
of which is shared between two players.  Each player receives a set of $O(n)$
edges, connecting the two sets of vertices shared with the other two players,
and a label in $\{0, 1\}$ for each edge.  These edges form $n$ disjoint
triangles, with each player having one edge from each triangle; every other
edge is isolated.  The players do not know which of their edges are in
triangles.  The promise is that for every triangle, the XOR of the associated
bits has the same value $\tau \in \{0, 1\}$; the goal is to find $\tau$.
In~\cite{KKP18} this was shown to take $\Omega(n^{1/3})$ bits of communication
in the SMP model.

\begin{figure}
  \centering
  \begin{subfigure}[t]{0.47\textwidth}
    \centering
    \includegraphics[width=\textwidth]{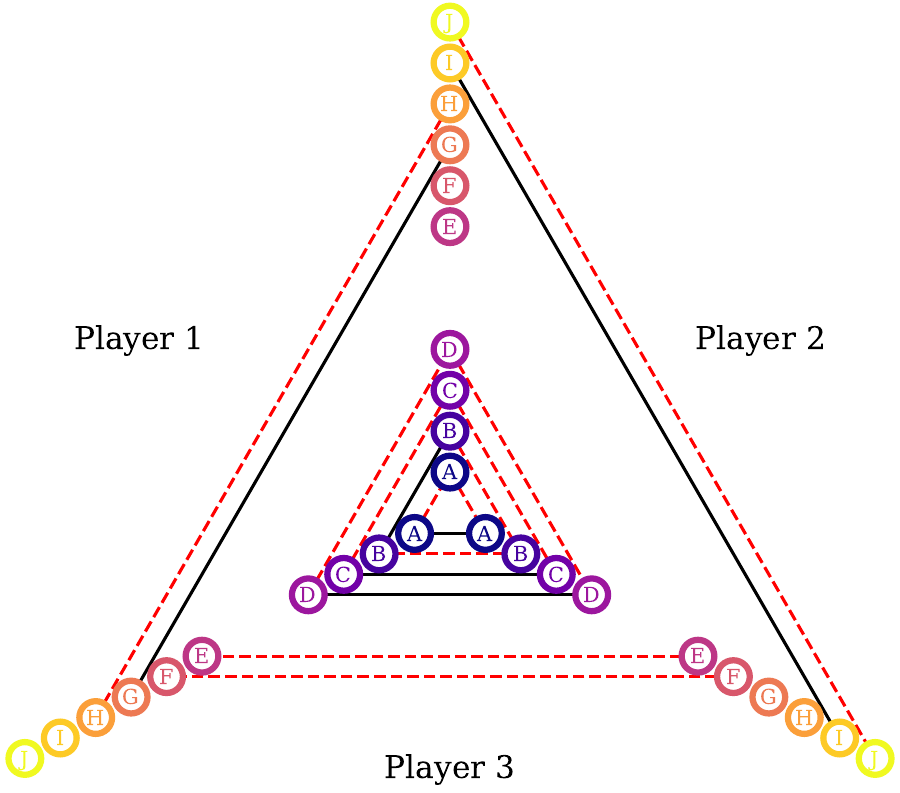}
    \caption{The player's instance ignoring the permutations.  The
      $x_e$ are the indices of red edges, read from inside out:
      $x_1 = [1,0,1,1,0,1]$, $x_2 = [1,1,1,1,0,1]$, $x_3 = [0,1,0,0,1,1]$.}
    \label{fig:fig2a}
  \end{subfigure}\hfill
  \begin{subfigure}[t]{0.47\textwidth}
    \centering
    \includegraphics[width=\textwidth]{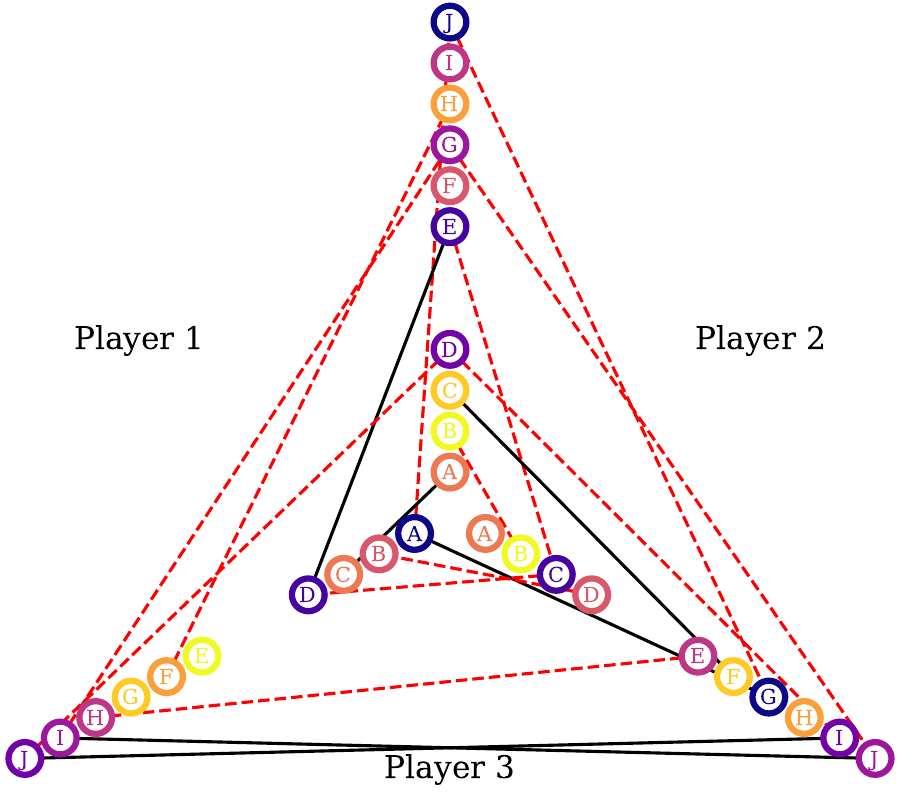}
    \caption{The hard distribution permutes each set of vertices.
      The players see their edges and associated labels, but not the
      vertex colors (which represent the pre-permutation
      identities).}
    \label{fig:fig2b}
  \end{subfigure}

  \vspace{1em}
  \begin{subfigure}[t]{0.47\textwidth}
    \centering
\begin{tabular}{|l|l|l||l|l|l||l|l|l|}
\hline
\multicolumn{3}{|c||}{Player 1} & \multicolumn{3}{|c||}{Player 2} & \multicolumn{3}{c|}{Player 3} \\
\hline
u&v&x&u&v&x&u&v&x\\
\hline\hline
A&J&1&B&B&1&C&D&1\\
\hline
C&A&0&C&F&0&D&B&1\\
\hline
D&E&0&D&I&1&E&H&1\\
\hline
F&H&1&E&C&1&G&A&0\\
\hline
I&G&1&G&J&1&I&J&0\\
\hline
J&D&1&J&G&1&J&I&0\\
\hline
\end{tabular}
    \caption{Each player's input consists of their edges
      in~(\subref{fig:fig2b}).  $u$ represents the
      vertex counterclockwise of the player, and $v$ represents the vertex clockwise.}
    \label{fig:fig2c}
  \end{subfigure}\hfill
  \begin{subfigure}[t]{0.47\textwidth}
    \centering
    \renewcommand{\arraystretch}{1.5}
    \begin{tabular}{r|cc}
      &$u \to v$ & $v \to u$\\
      \hline
Player 1 & $\overline{J}\bot{}AE\bot{}\overline{H}\bot{}\bot{}\overline{G}\overline{D}$ & $C\bot{}\bot{}\overline{J}D\bot{}\overline{I}\overline{F}\bot{}\overline{A}$\\
Player 2 & $\bot{}\overline{B}F\overline{I}\overline{C}\bot{}\overline{J}\bot{}\bot{}\overline{G}$ & $\bot{}\overline{B}\overline{E}\bot{}\bot{}C\overline{J}\bot{}\overline{D}\overline{G}$\\
Player 3 & $\bot{}\bot{}\overline{D}\overline{B}\overline{H}\bot{}A\bot{}JI$ & $G\overline{D}\bot{}\overline{C}\bot{}\bot{}\bot{}\overline{E}JI$
    \end{tabular}
    \caption{The encoding into $\Sigma^{6n}$.  For Theorem~\ref{thm:01-clean},
    each character in $\Sigma$ is encoded into binary; for
    Theorem~\ref{thm:2m-clean}, the encoding is instead in $\lbrace -m,
    m\rbrace$.}
    \label{fig:fig2d}
  \end{subfigure}
  
  \caption{Illustration of $\promise(4)$ instance; the true instance would have 36 isolated edges per player, not 2.}\label{fig:instance2}
\end{figure}

Each player's input can be represented in $k = O(n \log n)$ bits.  We can
define a data stream problem $P \subset \{0, 1\}^{3k} \times \{0, 1\}$ as
follows: for any input $x \in \{0, 1\}^{3k}$, split $x$ into three pieces $x^A,
x^B, x^C$, one for each player.  If $(x^A, x^B, x^C)$ represents a valid set of
inputs to $\promise(n)$, let $\tau$ be the corresponding answer and place $(x,
\tau)$ in $P$; otherwise, place both $(x, 0)$ and $(x, 1)$ in $P$.  Since the
players' inputs are placed in separate coordinates, a linear sketch could solve
the SMP communication problem, giving an $\Omega(n^{1/3}) = \Omega(k^{1/3} /
\log^{1/3} k)$ lower bound for linear sketches.  But how can we solve this
problem more efficiently with an arbitrary turnstile streaming algorithm?

The lower bound in \cite{KKP18} can be seen as proving that optimal algorithms
for this problem in the SMP model must be based on \emph{sampling}, where the
players each choose a subset of their edges/bit labels to send, and succeed if
there is some triangle such that each of the three players choose the edge they
hold from it. What makes the problem hard, then, is the fact that it is
difficult for all three players to simultaneously coordinate their sampling.
Any two players can coordinate: they can use shared randomness to sample a
shared vertex, and each keep their edge incident to that vertex.  But they
can't tell the third player which edge to keep.  

The idea behind our algorithm is that for any stream, for each triangle some
player's input will finish updating last.  As soon as the first two players'
inputs have finished updating, the algorithm will know which of their edges it
sampled, and therefore know what parts of the third player's input ar.  If the
third player's input hasn't finished yet, the algorithm will learn at least one
bit when it is updated.  And to solve $\promise(n)$, we only need one bit.

For this to work, we need an encoding of the players' inputs that satisfies a
few properties.  We need to be able to sample a vertex, and learn the incident
edges if we pay attention for the whole stream. If this vertex is incident to
two edges of a triangle, then once we learn one of these edges, we need to know
where in the vector to find the encoding of the third edge, and if we learn at
least one bit of the third edge's encoding, we need to be able to be able to
compute its bit label $z$ at the end of the stream.  This last point might seem
tricky, but at the end of the stream the sampled edges tell us both endpoints
of the third edge, so $z$ is the only bit we don't know; it will therefore
suffice to store an edge $(u, v, z)$ as $(u \oplus z^{B}, v \oplus z^{B})$
for a slightly larger word size $B$.  The precise encoding and recovery
algorithm are presented in Sections~\ref{sec:boxencoding} and~\ref{sec:01alg},
respectively.

\paragraph{Box-constrained streams.}  For Theorem~\ref{thm:2m-clean},
we take the same instance as for binary streams but place it on
$\{-M, M\}^{3k}$.  It is no longer the case that, once we start
tracking a given coordinate, we can learn its value after a single
update.  But we can still track the coordinate relative to its initial
value, and if the coordinate's final value is $M$ more than the smallest value
seen, or $M$ less than the largest value seen, then we \emph{will} know the
coordinate's value at the end of the stream, as there will be only one of
$\{-M, M\}$ for which this is consistent with staying within $\Gamma_{2M-1}$.

Now, optimistically decoding based on the sign pattern of each word, we define
the `last' player for a triangle as being the player whose input's decoding
achieves its final value last, i.e.\ the last player to have every coordinate
of their input within $M-1$ of its final value.  At the time the first two
players' inputs' decodings achieve their final value, these players will know
their sampled edges, and there will be at least one coordinate of the third
player's input that can be learned with the remaining stream.

\subsubsection{Bounded Degree Triangle Counting}

At a high level, both of our algorithms for bounded-degree triangle counting
seek to emulate the insertion-only algorithm of~\cite{JG05}.  The
insertion-only algorithm is as follows: sample edges with probability $p$, and
keep all edges incident to sampled edges.  Count the number of triangles using
sampled edges (with multiplicity if multiple edges of a triangle are sampled),
and divide by $p$.  This is an unbiased estimator, using $O(pmd \log n)$
space, in a graph with $m$ edges, $n$ vertices, and max degree $d$.  The
expected number of triangles sampled is $pT$.  If all the triangles were
disjoint, the triangles would be sampled independently and so one could set $p
= O(1/(\eps^2 T))$ and get a $(1 + \eps)$-approximation with $2/3$ probability.
Even though the triangles are not disjoint, the degree bound keeps the
estimator's variance small; one only needs $p = O(d/(\eps^2 T))$.

So what happens in turnstile streams?  One can run essentially the same
algorithm, dealing with edge deletions by removing both the edge deleted and
any neighbors that were tracked on its account.  This works, but can use too
much space if not done carefully.

\paragraph{Bounded-degree intermediate states.}  If every intermediate state is
a bounded-degree graph, then the expected amount of space used at any point in
the stream is still $O(pmd \log n)$.  However, if the stream is extremely long,
the \emph{maximum} amount of space used will be too large.  The natural
solution is to have a hard cap of $O(pm)$ on the number of edges sampled, and
to stop sampling edges when at the cap.  One might worry that this creates a
bias in the estimator.  However, the only times this can affect the output of
the algorithm are the $m$ points in time when edges in the final graph are
inserted for the last time.  At each such time, with high probability, the hard
cap will not have been reached.  The output of the algorithm will thus be the
same as in the insertion-only case.

\paragraph{Length-constrained streams.}  In this model, the intermediate states
may be multigraphs with very high degree; call the maximum degree a vertex ever
reaches its `stream degree.'  One cannot, in general, keep the entire
neighborhood of a sampled edge.  However, the $\Omega(T/d)$ edges involved in
triangles in the final graph have average stream degree at most
$O(\frac{Ld}{T})$.  Therefore we can restrict to considering edges of stream
degree $O(\frac{Ld^2}{\eps T})$: this loses us at most an $\eps/3d$ fraction of
triangle-involved edges, which are involved in at most an $\eps$ fraction of
triangles.

Using the same $p = O(d/(\eps^2 T))$ as in the insertion-only case, we
get an algorithm with space
\[
  p \cdot L \cdot \frac{Ld^2}{\eps T} \cdot \log n = O(\frac{d^3L^2}{\eps^3 T^2} \log n).
\]

\subsection{Deterministic Turnstile-Sketching Equivalence}
Our strategy for reducing deterministic turnstile streaming to linear sketching
will be to take a turnstile streaming algorithm and give it various
streams as input until we find vectors that can be safely ``quotiented out''.
By repeatedly doing this we can find a linear map (a homomorphism of
$\Z$-modules) from $\Z^n$ to a module of size at most $2^s$, whose elements can
be represented as sparse vectors in $\Z^n$.

In each case, the existence of these vectors will be guaranteed by the fact
that $\alg$ can have at most $2^s$ different states, and we will be able to
find them by looking for ``collisions'' in these states---streams which result
in different vectors but the same state of $\alg$. How we find them, and the
length of streams we will need $\alg$ to tolerate, will depend on whether $\alg$
calculates some total function on $\Z^n$ exactly, or whether it solves a general
``streaming problem''---that is, each input has multiple valid outputs, e.g., a
counting problem where only $(1 \pm \varepsilon)$ multiplicative accuracy is 
needed.

\tikzset{basevertex/.style={shape=circle, line width=0.5,
minimum size=4pt, inner sep=0pt, draw}}
\tikzset{defaultvertex/.style={basevertex, fill=blue!70}}
\begin{figure}
	\label{fig:totalfunctions}
	\centering
	\begin{tikzpicture}[%
			VertexStyle/.style={defaultvertex},
			fat arrow/.style={single arrow,
			thick,draw=blue!70,fill=blue!30,
			minimum height=10mm},
			scale = 1.4]
			\SetVertexNoLabel
			\Vertex[x=0,y=2]{x0}
			\Vertex[x=1,y=2.5]{x1}
			\Vertex[x=2,y=2.1]{x2}
			\Vertex[x=3,y=1.8]{x3}
			\tikzset{defaultvertex/.style={basevertex, fill=red!70}}
			\Vertex[x=0,y=0]{y0}
			\Vertex[x=1,y=0]{y1}
			\Vertex[x=2,y=0.1]{y2}
			\Vertex[x=3,y=0.7]{y3}
			\tikzset{defaultvertex/.style={basevertex, fill=violet!70}}
			\Vertex[x=-1,y=1]{origin}
			\Vertex[x=4,y=1]{xy4}

			\node at (-1.4, 1) {$\alg(\emptyset)$};

			\draw[-stealth] (origin) -- (x0);
			\draw[-stealth] (x0) -- (x1);
			\draw[-stealth] (x1) -- (x2);
			\draw[-stealth] (x2) -- (x3);
			\draw[-stealth] (x3) -- (xy4);

			\node at (1.5, 2.5) {$\kappa(x)$};

			\draw[-stealth] (origin) -- (y0);
			\draw[-stealth] (y0) -- (y1);
			\draw[-stealth] (y1) -- (y2);
			\draw[-stealth] (y2) -- (y3);
			\draw[-stealth] (y3) -- (xy4);

			\node at (1.5, -0.3) {$\kappa(y)$};

      \node at (4.2,1.3) {$\alg(\kappa(x))$};
      \node at (4.2,0.7) {$\alg(\kappa(y))$};

            \Vertex[x=5,y=1]{xy5}
            \Vertex[x=6,y=1]{xy6}
            \Vertex[x=7,y=1]{xy7}
            \Vertex[x=8,y=1]{xy8}
            \Vertex[x=9,y=1]{xy9}

            \draw[-stealth] (xy4) -- (xy5);
            \draw[-stealth] (xy5) -- (xy6);
            \draw[-stealth] (xy6) -- (xy7);
            \draw[-stealth] (xy7) -- (xy8);
            \draw[-stealth] (xy8) -- (xy9);

            \draw [decorate,decoration={brace,amplitude=10pt,raise=1cm},]
            (xy4) -- (xy9) node [black,midway,yshift=1.6cm] 
            {$\kappa(z - y)$};


            \node at (9,1.3) {$\alg(\kappa(x)\cdot\kappa(z-y))$};
            \node at (9,0.7) {$\alg(\kappa(y)\cdot\kappa(z-y))$};

	\end{tikzpicture}
   \caption{For total functions, we need to find pairs of streams which cause
   $\alg$ to reach the same state. Here we find $x$ and $y$ such that their
   canonical representations $\kappa(x)$ and $\kappa(y)$ reach the same state.
   This means that for \emph{any} $z$ there are streams with frequency $z$ and
   $z + (x - y)$ that reach the same state, so $f(z) = f(z + (x-y))$.}
\end{figure}
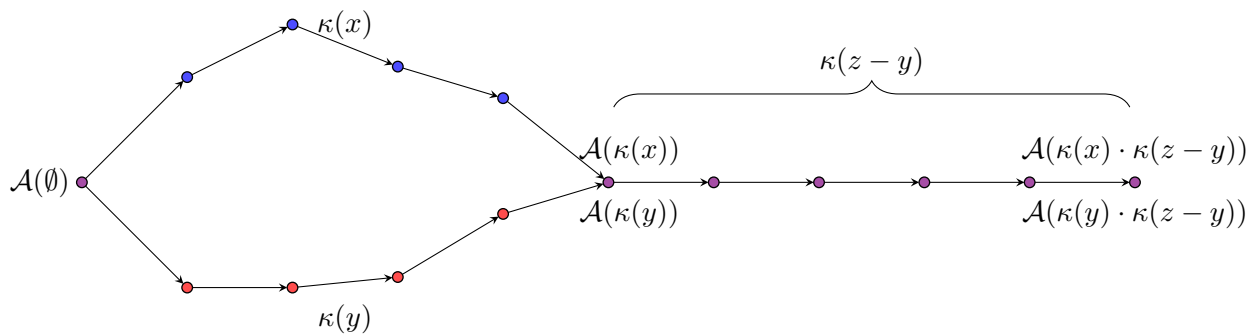

\begin{figure}
  \label{fig:partialfunctions}
  \centering
  \begin{tikzpicture}[%
      VertexStyle/.style={defaultvertex},
      fat arrow/.style={single arrow,
      thick,draw=blue!70,fill=blue!30,
      minimum height=10mm},
      scale = 1.4]
      \SetVertexNoLabel
    \Vertex[x=1,y=0] {x1}
    \Vertex[x=2,y=0] {x2}
		\tikzset{defaultvertex/.style={basevertex, fill=red!70}}
    \Vertex[x=3.866,y=0.5] {x4}
    \Vertex[x=3.866,y=1.5] {x5}
    \Vertex[x=3,y=2] {x6}
    \Vertex[x=2.134,y=1.5] {x7}
    \Vertex[x=2.134,y=0.5] {x8}
		\tikzset{defaultvertex/.style={basevertex, fill=violet!70}}
    \Vertex[x=3,y=0] {x3}

    \draw[-stealth] (x1) -- (x2);
    \draw[-stealth] (x2) -- (x3);
    \draw[-stealth] (x3) -- (x4);
    \draw[-stealth] (x4) -- (x5);
    \draw[-stealth] (x5) -- (x6);
    \draw[-stealth] (x6) -- (x7);
    \draw[-stealth] (x7) -- (x8);
    \draw[-stealth] (x8) -- (x3);
    
    \node at (0.6,0) {$\alg(\emptyset)$};
    \node at (3,0.5) {$\alg(\pi_1\cdot \rho_1)$};
    \node at (3,-0.35) {$\alg(\pi_1)$};

    \tikzset{defaultvertex/.style={basevertex, fill=blue!70}}
    \Vertex[x=4,y=0] {x9}
    \Vertex[x=5,y=0] {x10}
    \tikzset{defaultvertex/.style={basevertex, fill=red!70}}
    \Vertex[x=6.866,y=0.5] {x12}
    \Vertex[x=6.866,y=1.5] {x13}
    \Vertex[x=6,y=2] {x14}
    \Vertex[x=5.134,y=1.5] {x15}
    \Vertex[x=5.134,y=0.5] {x16}
    \tikzset{defaultvertex/.style={basevertex, fill=violet!70}}
    \Vertex[x=6,y=0] {x11}
 
    \draw[-stealth] (x3) -- (x9);
    \draw[-stealth] (x9) -- (x10);
    \draw[-stealth] (x10) -- (x11);
    \draw[-stealth] (x11) -- (x12);
    \draw[-stealth] (x12) -- (x13);
    \draw[-stealth] (x13) -- (x14);
    \draw[-stealth] (x14) -- (x15);
    \draw[-stealth] (x15) -- (x16);
    \draw[-stealth] (x16) -- (x11);

    \node at (6,0.5) {$\alg(\pi_2\cdot \rho_2)$};
    \node at (6,-0.35) {$\alg(\pi_2)$};

    \draw [decorate,decoration={brace,amplitude=10pt,raise=1cm},]
    (x11) -- (x1) node [black,midway,yshift=-1.6cm] 
    {$\pi_2 = \tau_2^{x_{j'}} = \pi_1 \cdot \psi_2$};

    \node[above=.3cm of x14] (label) {Loop $i$ has frequency $a_i e_i - o_i$};

    \tikzset{defaultvertex/.style={basevertex, fill=blue!70}}
    \Vertex[x=7,y=0] {x17}
    \Vertex[x=8,y=0] {x18}
    \tikzset{defaultvertex/.style={basevertex, fill=red!70}}
    \Vertex[x=9.866,y=0.5] {x20}
    \Vertex[x=9.866,y=1.5] {x21}
    \Vertex[x=9,y=2] {x22}
    \Vertex[x=8.134,y=1.5] {x23}
    \Vertex[x=8.134,y=0.5] {x24}
    \tikzset{defaultvertex/.style={basevertex, fill=violet!70}}
    \Vertex[x=9,y=0] {x19}

    \draw[-stealth] (x11) -- (x17);
    \draw[-stealth] (x17) -- (x18);
    \draw[-stealth] (x18) -- (x19);
    \draw[-stealth] (x19) -- (x20);
    \draw[-stealth] (x20) -- (x21);
    \draw[-stealth] (x21) -- (x22);
    \draw[-stealth] (x22) -- (x23);
    \draw[-stealth] (x23) -- (x24);
    \draw[-stealth] (x24) -- (x19);

    \node at (9,0.5) {$\alg(\pi_3\cdot \rho_3)$};
    \node at (9,-0.35) {$\alg(\pi_3)$};

    \tikzset{defaultvertex/.style={basevertex, fill=blue!70}}
    \Vertex[x=10,y=0] {x25}
    \Vertex[x=11,y=0] {x26}

    \draw[-stealth] (x19) -- (x25);
    \draw[-stealth] (x25) -- (x26);

    \node at (10,-0.4) {$\sigma$};

  \end{tikzpicture}
  \caption{For general streaming problems, we generate one very long stream
  which iterates though a sequence of vectors $x_i$ in $\Z^n$, looking for
  ``loops'' that change the value of the vector without changing the state of
  $\alg$. Whatever the postfix $\sigma$ is, the output will be indifferent to
  the number of loops $\rho_i$ added.}
\end{figure}
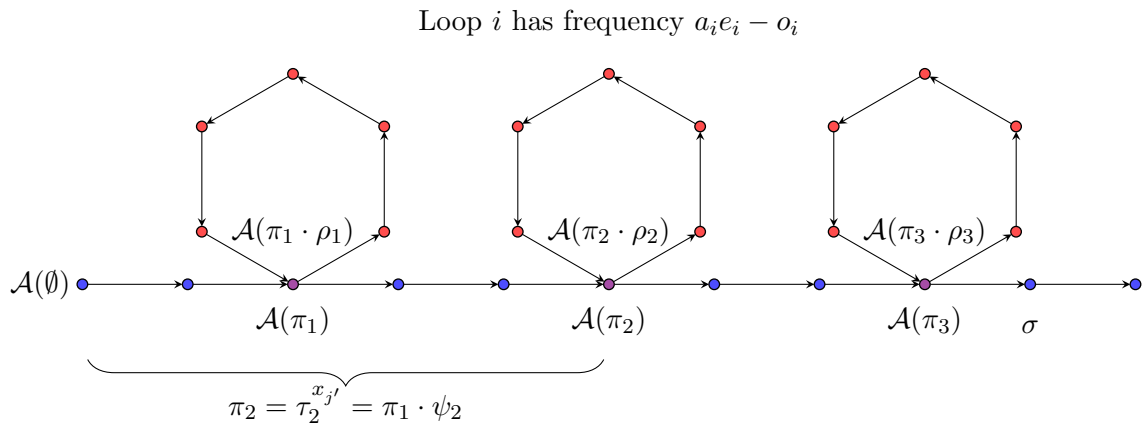

\paragraph{Total functions.} For total functions $f$, we will consider streams
that are the ``canonical representation'' $\kappa(x)$ of some vector $x$,
defined as the stream that inserts every coordinate of $x$. If we can find
some pair of vectors $x, y$ such that the algorithm reaches the same state on
$\kappa(x)$ and $\kappa(y)$, then for \emph{any} vector $z$, the algorithm will
reach the same result on $\kappa(x) \cdot \kappa(-y) \cdot \kappa(z)$ and
$\kappa(y) \cdot \kappa(-y) \cdot \kappa(z)$, and so $f(z) = f(z + (x-y))$. It
is therefore safe to ``quotient'' out $x-y$. 

By repeatedly performing this procedure, we find a submodule of $\Z^n$ such
that $f$ is constant on the submodule and all its cosets---our sketch can
be seen as a map from $\Z^n$ to the corresponding quotient module.

As any vector in $\Z^n$ can be inserted in at most $n$ updates, this
means we only need $\alg$ to work on length $O(n)$ streams. In fact,
it will prove possible to guarantee $x$ and $y$ are length no more
than $s$, so provided $\alg$ is sublinear the required stream length
is $(1 + o(1))n$.

At the end of the stream, having stored a ``reduced'' vector, we recover $f$ by
presenting this vector to $\alg$ in its canonical form---as we know $f$ takes 
the same value on the reduced vector as it does on the full input vector we 
will recover the correct answer.

\paragraph{General streaming problems.} The above approach fails, however, if
$\alg$ has multiple valid outputs for any given input. To see this, consider
the case where $\alg$ calculates a $(1 \pm \varepsilon)$ approximation to $f$.
Then the proof above would guarantee only that $f(z)$ and $f(z + (x-y))$ were
within $\varepsilon(f(z) + f(z + (x-y))$ of one another, and so repeatedly
quotienting out vectors could still bring us very far from the correct answer.

So instead of finding a submodule such that $f$ is constant on cosets of the
submodule, we find a submodule such that there is a mapping from vectors in
$\Z^n$ to streams such that for each coset of the submodule, the output of
$\alg$ on the corresponding streams is constant. We can then quotient out
the vectors that generate this submodule, and then once we are finished
processing the stream, map our ``reduced'' vector to an appropriate stream and
give that stream as input to $\alg$.

To do so we will consider a sequence of vectors $x_i$ that iterates through
$\Z$ in some appropriate way, and the corresponding ``covering streams''
$\tau^{x_i} = \kappa(x_1)\cdot\kappa(x_2 - x_1) \dots \kappa(x_i - x_{i-1})$.
As $\alg$ only has $2^s$ states, at some point when processing this stream it
will return to a state already visited. This gives us a ``loop'', a sequence of
updates that takes us from one state to the same state. As the $x_i$ are
distinct, we can find a loop that has non-zero frequency, and therefore we can
quotient out that loop. 

We repeat this process to find a sequence of streams $\pi_i$ (each a prefix of
the next) and loops $\rho_i$ such that the algorithm is the same after
processing $\pi_i \cdot \rho_i$ as $\pi_i$, but $\freq \rho_i$ is a different
non-zero vector each time

For recovery, we will again insert the reduced vector in its canonical form in
$\alg$, but we will need to prefix it with the stream built up in the reduction
(without loops). We then subtract off the original stream to preserve the final
value of the vector. That ensures that there is some stream which corresponds
to the original vector such that $\alg$ would reach the same state it does on
this one (by inserting loops\footnote{It may be noted that this will not work
if taking the original vector to the reduced vector requires \emph{subtracting}
our ``quotiented out'' vectors. To compensate for this, our mapping from
vectors to streams will include subtracting a large number of each quotient
vector (outside of the loops), so that we only need to add loops. It is
possible to show that there is a sufficiently large number of
quotients to subtract independent of the true value of the vector.}), and so
whatever output our algorithm gives is
some valid output for this vector.

\paragraph{Constructing a sketch.} In both cases, we have described a
method of finding vectors to ``reduce'' our input vector by---in other
words, we have found a way to produce vectors that generate a
submodule $N$ of $\Z^n$ such that we only care which coset of $N$ our
vector is in (i.e. which element of the quotient module $\Z^n/N$ it
maps to). However, we still need to find a consistent method to take
an element $x$ of $\Z^n$ to a representative element of $N + x$ that
can be computed in small space. Moreover, we need to be able, for any
pair of representative elements $x, y$ to find the representative
element of $N + (x + y)$, so that we can apply module operations
(i.e., maintain the sketch under updates to the stream and merge
sketches of different streams).

The representative element we choose is the lexicographically first element
with all non-negative coordinates in $N + x$. This can be computed in small
space by repeatedly subtracting off our ``quotient vectors'' until it is no
longer possible to do so (we will choose these vectors in a way that guarantees
this eventually happens). The set of these elements will turn out to be
$\prod_{i=1}^n\Z_{a_i}$ for positive integers $a_i$, and we will call the map
from $\Z^n$ to this set $\phi$. For any pair of representative elements $x, y$,
the representative element of $N + (x + y)$ will be $\phi(x + y)$, so this
defines a $\Z$-module $M \cong \Z^n/N$ with addition operator $\star$ given by
$x \star y = \phi(x + y)$ and $\phi$ a homomorphism between these modules.

To actually calculate this homomorphism, we need to calculate the vectors to
be quotiented out in $O(s \log n)$ space. As even storing all of them would
require more space than that, we generate them sequentially whenever needed,
storing only enough information about vectors generated earlier to calculate
later vectors.

The proof of these results lies in Section~\ref{sec:equivalence}.

\section{Box-Constrained Streaming: Problem and lower bound}\label{sec:boxencoding}
\subsection{Streaming Triangle Game}

\noindent
Our problem is based on encoding an instance of the
$\mathtt{PromiseCounting}(H,n,T,\eps)$ communication problem from
\cite{KKP18} as a binary vector.  We will only
use the special case where $H$ is the triangle $K_3$, $T = n/10$,
and $\eps = 1$.  We refer to this
$\mathtt{PromiseCounting}(K_3,n,n/10,1)$ instance as $\promise(n)$,
which we describe in Figure~\ref{fig:promise} and illustrate in
Figure~\ref{fig:instance2}.

\begin{figure}
\begin{center}
\fbox{
\begin{minipage}{\textwidth}
{
\noindent $\promise(n)$

\paragraph{Parties:} Let $\Vt$ and $\Et$ be the vertex and edge sets, respectively, of a triangle $K_3$.  There are three players, one associated with each edge $e \in \Et$.  There is one referee, who receives messages from the three players. No other communication takes place.

\paragraph{Constants:} Let $N = 30n$.  We define 
$N$ vertices $V_a$ associated with each of the three vertices $a \in V^\Delta$.

\paragraph{Inputs:}

Each player $e = ab$ receives a list of $N/3$ triples $(u, v, z_{uv}) \in V_a \times V_b \times \{0, 1\}$.

\paragraph{Promise:}
The instance satisfies the following promise:
\begin{enumerate}
\item No $u$ or $v$ appears more than once in any single player's input.  Thus the set of all edges $(u, v)$ in player inputs can be viewed as a graph $G$ over $\bigcup_{a \in \Vt} V_a$, and this graph has $N$ edges and $3N$ vertices.
\item $G$ contains $n$ triangles.  All $27 n$ other edges are isolated.
\item There exists a $\tau \in \{0, 1\}$ such that for every triangle $uvw$ in $G$,
  \[
    z_{uv} \oplus z_{vw} \oplus z_{wu} = \tau.
  \]
\end{enumerate}

\paragraph{Goal:}
Given the messages received from the players, the referee's task is to
determine whether $\tau = 0$ or $\tau = 1$.  }
\end{minipage}
}
\end{center}
\caption{Definition of a $\promise$ instance.}\label{fig:promise}
\end{figure}

\begin{theorem}[Implication of Corollary 15 of~\cite{KKP18}]\label{thm:KKP}
  Let $n \geq 1$.  Suppose that, for every instance of $\promise(n)$,
  no player sends a message of more than $c$ bits.  There exists
  a universal constant $\gamma$ such that, if $c \leq \gamma n^{1/3}$, the
  probability the referee succeeds is at most 51\%.
\end{theorem}

We note that our~$\promise$ problem is written somewhat differently
from \linebreak the~$\mathtt{PromiseCounting}$ problem as defined
in~\cite{KKP18}.  Our description is equivalent, however, as suggested
in Figure~2 of~\cite{KKP18}.

Both Theorem~\ref{thm:01-clean} and Theorem~\ref{thm:2m-clean} involve encoding
the player's inputs to $\promise(n)$ as a frequency vector. The outer encoding,
from instances of $\promise(n)$ to strings from an alphabet $\Sigma$, is the
same for both. The inner encoding will differ, taking strings from $\Sigma$ to
strings from $\lbrace 0, 1\rbrace$ and $\lbrace -M, M\rbrace$ for
Theorem~\ref{thm:01-clean} and Theorem~\ref{thm:2m-clean} respectively.

For both, the frequency vector will have dimension $\Theta(n \log n)$.
Theorems~\ref{thm:01-clean} and~\ref{thm:2m-clean} then follow by considering
an encoding of $\promise(\Theta(n/\log n))$.

$\promise$ is defined in Figure~\ref{fig:promise}. When there is no ambiguity
about which instance of~$\promise$ is being referenced, we will implicitly use
the variable names from this definition to refer to the corresponding variables
for that instance.

\paragraph{Outer Encoding.} We define the alphabet
$\Sigma = ([N] \times \{0, 1\}) \cup \{\perp\}$.  We encode an
instance of $\promise(n)$ into $\Sigma^{6N}$ as follows.  For each
$e \in \Et$ and $a \in e$, we create a vector
$y^{e,a} \in \Sigma^{N}$; the full encoding is the concatenation of
the six $y^{e,a}$.

As illustrated in Figure~\ref{fig:fig2c}, the input of player $e = ab$
consists of a list of $N/3$ edges $(u, v, z_{uv})$, where each
$u \in V_a$ and $v \in V_b$.  Since $\abs{V_a} = \abs{V_b} = N$, we
can define a canonical bijection from each of $V_a$ and $V_b$ into
$[N]$; call these $f_a$, $f_b$.

Then for every $(u, v, z_{uv})$ in player $e$'s list, we set
\begin{align*}
  y^{e, a}_{f_a(u)} &:= (f_b(v), z_{uv})\\
  y^{e, b}_{f_b(v)} &:= (f_a(u), z_{uv})
\end{align*}
Since each $u$ appears at most once in $e$'s list, this is well defined.  This
sets $N/3$ of the $N$ coordinates in each of $y^{e,a}$ and $y^{e,b}$; every
other coordinate is set to $\perp$.  

This encoding of the players' inputs is injective; in fact, either one
of $y^{e,a}$ or $y^{e,b}$ suffices to recover player $e$'s input.

\paragraph{Inner Encoding.}  Let $B = 1 + \lceil \lg N + 1 \rceil$.
For Theorem~\ref{thm:01-clean}, we encode $\Sigma$ into
$\{0, 1\}^{B}$.  We encode $\perp$ as $0^B$.  To encode
$(l, z) \in [N] \times \{0, 1\}$ we first take the standard binary
encoding $l^{(bin)}$ of $l$ into $\{0, 1\}^{B}$.  This is nonzero,
since $l > 0$; and its highest bit is zero, since $l \leq N$.  Then we
output the bitwise XOR $x = l^{(bin)} \oplus z^B$.

This encoding is injective, because the highest bit will equal $z$,
after which $z$ can be removed and $l$ recovered.  Concatenating the
outer and inner code gives an injection from the players' inputs to $\{0,
1\}^{6NB}$.

For Theorem~\ref{thm:2m-clean}, we use the same encoding, and then replace
every instance of $1$ with $M$, and every instance of $0$ with $-M$.

\paragraph{The streaming problem.}  We can now define the streaming
problem $P_n$. For any vector $x$ such that $x$ is \emph{not} an
encoding of an instance of $\promise(n)$, $(x,0)$ and $(x,1)$ are in
$P_n$, i.e., any output is acceptable on such an input. For any vector
$x$ such that $x$ is an encoding of an instance with $\tau = 0$,
$(x, 0) \in P_n$, and for any vector $x$ such that $x$ is an encoding
of an instance with $\tau = 1$, $(x, 1) \in P_n$.

\subsection{Linear Sketching Lower Bound}
By Theorem~\ref{thm:KKP}, any protocol for the communication problem that
succeeds with probability at least $2/3$ requires $\Omega(n^{1/3})$ bits of
communication by at least one player. Furthermore, the model of~\cite{KKP18}
allows the players access to an unlimited amount of shared randomness.

Now suppose we have a linear sketching algorithm for $P_n$.  Note that the
outer code encodes each player's input into separate coordinates.  The inner
code, of course, preserves this property.  Therefore player $e$ could encode
their part of the problem with the other coordinates set to zero, sketch it,
and send it to the referee.  The referee can add up these sketches to get a
sketch for the full vector $x$, then determine $\tau$.  Since each player only
sends a message of size equal to the space usage of the linear sketching
algorithm, the space used must be $\Omega(n^{1/3})$.

Therefore, $P_n$ satisfies criterion 1 of Theorems~\ref{thm:01-clean}
and~\ref{thm:2m-clean}. To prove that it satisfies criterion 2, we
construct a turnstile algorithm for $P_n$.

\subsection{Algorithm for $\promise$ over $\Gamma_{0,1}$}\label{sec:01alg}

This section will describe an algorithm that either outputs the correct answer
or $\perp$, and outputs the correct answer with a small positive constant
probability.  Straightforward probability amplification then can increase the
success probability to $2/3$.

We start by noting that, for any coordinate $i$, we can establish $x_i$ given
any non-empty postfix of the updates to $x_i$, as any increase proves it was
previously 0 and any decrease proves it was previously 1.

Recall that any player $e \in \Et$, side $a \in e$, and vertex
$u \in V_a$ has an associated symbol $y^{e,a}_{f_a(u)} \in \Sigma$.
We use $x^{e, a, u} \in \{0, 1\}^B$ to denote the inner encoding of
this symbol.  The final frequency vector $x$ has $x^{e, a, u}$
placed in a contiguous block, at a position that is easy to find from
$(e, a, u)$.

We state the algorithm in Algorithm~\ref{alg:weak01}.

\RestyleAlgo{boxruled}
\begin{algorithm}\caption{Low-probability $\promise$ over $\{0, 1\}$}\label{alg:weak01}
\addtolength\linewidth{-4.6ex}

\begin{enumerate}
\item Let $(a, b, c)$ be a uniformly chosen random labeling of $\Vt$.  Choose
  $u \in V_a$ uniformly at random.
\item While passing through the stream:
  \begin{enumerate}
  \item Track all updates to $x^{ab,a,u}$ and $x^{ac, a, u}$.
  \item While doing so, keep checking whether $x^{ab,a,u}$ is a valid
    inner encoding of $\Sigma$; if it is, and it doesn't decode to $\perp$, then it is
    an encoding of $(f_b(v'), z)$ for some $v' \in V_b$ and $z'$.  Let $(v', z')$
    be those values, if they exist.
  \item As soon as $(v', z')$ is set, track all updates to $x^{bc,b,v'}$.
    Discard these updates whenever $(v', z)$ changes.
\end{enumerate}
\item After the stream finishes:
  \begin{enumerate}
  \item Decode $x^{ab,a,u}$ and $x^{ac, a, u}$ to $\Sigma$.
  \item If either is $\perp$, output $\perp$.
  \item Otherwise, let their decodings be $(f_b(v), z_{uv})$ and
    $(f_c(w), z_{uw})$ for $v \in V_b$ and $w \in V_c$.
  \item If the algorithm has not tracked any updates to $x^{bc,b,v}$, output $\perp$.
  \item Otherwise, it knows $x^{bc,b,v}_i$ for some index $i \in [B]$.
    Let $z_{vw} = x^{bc,b,v}_i \oplus f_c(w)^{(bin)}_i$.
  \item Output $z_{uv} \oplus z_{vw} \oplus z_{uw}$.
  \end{enumerate}
\end{enumerate}
\end{algorithm}
\begin{lemma}
The space complexity of Algorithm~\ref{alg:weak01} is $O(\log n)$ bits.
\end{lemma}
\begin{proof}
  The randomness in step 1 uses $\log (6N)$ bits.  After that, the algorithm
  tracks three length-$B$ vectors; the total space usage is $O(\log n)$.
\end{proof}

\begin{lemma}\label{lem:weak01correct}
  Algorithm~\ref{alg:weak01} outputs either $\perp$ or $\tau$.  If $u$
  is part of a triangle in the underlying $\promise(n)$ graph $G$, and
  the last stream update to $x^{ab, a, u}$ is before the last stream
  update to $x^{bc,b,v}$, then the algorithm outputs $\tau$.
\end{lemma}
\begin{proof}
  Note that $x^{ab,a,u}$ and $x^{ac, a, u}$ are tracked completely, so their
  final decodings into $\Sigma$ are correct.  If $u$ is not part of a triangle,
  at most one edge is incident to $u$ in the full graph $G$, so at least one of
  the decodings is $\perp$ and the algorithm returns $\perp$.

  Otherwise, if $u$ is part of a triangle, the algorithm correctly
  deduces $(v, z_{uv})$ and $(w, z_{uw})$.  If the algorithm has not
  seen an update to $x^{bc,b,v}$, it will output $\perp$; otherwise,
  since it tracks a postfix of the stream, it correctly identifies
  $x^{bc,b,v}_i$.  Since $uvw$ is a triangle, we know player $bc$ has
  the input $(v, w, z_{vw})$ for some $vw$, and the inner encoding is
  \[
    x^{bc,b,v}_i = z_{vw} \oplus f_c(w)^{(bin)}_i.
  \]
  Thus the algorithm correctly identifies $z_{vw}$, and the
  $\promise(n)$ promise says
  \[
    \tau = z_{uv} \oplus z_{vw} \oplus z_{uw}.
  \]
  Hence the algorithm outputs either $\perp$ or $\tau$.  Moreover, it
  will have deduced $v$ correctly upon the last update to
  $x^{ab,a,u}$; if this is before the last update to $x^{bc,b,v}$ then
  it will see at least one update there and output $\tau$.
\end{proof}

\begin{lemma}\label{lem:weak01chance}
  Algorithm~\ref{alg:weak01} outputs $\tau$ with at least
  $\frac{1}{180}$ probability.
\end{lemma}
\begin{proof}
  There is a $n/N = 1/30$ chance that $u$ lies in a triangle, independent of
  the choice of $(a, b, c)$. Furthermore, if it does, \emph{which} triangle
  it lies in is independent of the choice of $(a, b, c)$.

  Suppose $u$ lies in the triangle $uvw$ with $u \in V_{a'}, v \in V_{b'}, w
  \in V_{c'}$.  One of the three blocks
  \[
    x^{a'b', a', u},\qquad x^{b'c', b', v},\qquad x^{c'a', c', w}
  \]
  will be the first to finish being updated in the stream.  WLOG this is $a'$.
  Then Lemma~\ref{lem:weak01correct} says that if $(a, b, c) = (a', b', c')$,
  Algorithm~\ref{alg:weak01} will output $\tau$.  This choice happens with
  $1/6$ probability; combined with the $1/30$ chance that $u$ lies in a
  triangle, we get at least a $1/180$ chance of outputting $\tau$.
\end{proof}

\begin{lemma}
There is a turnstile streaming algorithm that solves $P_n$ on $\Gamma_{0,1}$
with probability $2/3$ using $O(\log n)$ bits of space.
\end{lemma}
\begin{proof}
  Run Algorithm~\ref{alg:weak01} in parallel $360$ times and output
  any non-$\perp$ result.  By Lemma~\ref{lem:weak01correct} any
  non-$\perp$ result will be correct.  By Lemma~\ref{lem:weak01chance}
  the failure probability is at most $(1 - 1/180)^{360} < 1/e^2 < 1/3$.
\end{proof}

\subsection{Algorithm for $\promise$ over $\Gamma_{2M-1}$}\label{sec:boxalg}
We write $\sigma^{(t)}$ for the prefix of $\sigma$ consisting of its first $t$
updates. Define the error correction function $\zeta$ by
\[
\zeta(z)_i = \begin{cases}
M &\mbox{$z_i > 0$}\\
-M &\mbox{$z_i < 0$}\\
0 &\mbox{$z_i = 0$}
\end{cases}
\]
and define the decoding function $\eta : \lbrace -M, M\rbrace^*\rightarrow
\lbrace 0, 1\rbrace$ by: \[
\eta(z)_i = \begin{cases}
1 &\mbox{$z_i = M$}\\
0 &\mbox{$z_i = -M$}
\end{cases}
\]
We will use the following decoding lemma in our algorithm:
\begin{lemma}
\label{lem:2mdecoding}
Let $\sigma$ be a stream in $\Gamma_{2M - 1}$ such that $\freq \sigma \in
\lbrace -M, M\rbrace^*$. Then for any $i$, and for any split of the stream
$\sigma = \sigma_1 \cdot \sigma_2$,
\begin{enumerate}
\item $\min_t (\freq \sigma_2^{(t)})_i \le (\freq \sigma_2)_i - M \Rightarrow
\eta(\freq \sigma)_i = 1$
\item $\max_t (\freq \sigma_2^{(t)})_i \ge (\freq \sigma_2)_i + M \Rightarrow
\eta(\freq \sigma)_i = 0$
\end{enumerate}
and one of these conditions holds iff $\exists t$ such that $\zeta(\freq
\sigma_1 \cdot \sigma_2^{(t)})_i \not = \zeta(\freq \sigma)_i$.
\end{lemma}
\begin{proof}
Suppose $\min_t (\freq \sigma_2^{(t)})_i \le (\freq \sigma_2)_i - M$.  Then if
$\eta(\freq \sigma)_i = 0$, $(\freq \sigma)_i = -M$. Let $t$ be a minimizer
of $(\freq \sigma_2^{(t)})_i$, so
\begin{align*}
(\freq \sigma^{(|\sigma_1| + t)})_i &= (\freq \sigma_1)_i + (\freq
\sigma_2^{(t)})_i\\
&\le (\freq \sigma_1)_i + (\freq \sigma_2)_i - M\\
&= (\freq \sigma)_i - M\\
&= -2M
\end{align*} but by the box constraint $(\freq \sigma^{(t)})_i \ge -2M + 1$,
giving a contradiction. So $\eta(\freq \sigma)_i = 1$.

Likewise, if $\max_t (\freq \sigma_2^{(t)})_i \ge (\freq \sigma_2)_i + M$,
there exists $t$ such that if $\eta(\freq \sigma)_i = 1$, $(\freq
\sigma^{(|\sigma_1| + t)})_i \ge 2M$, so it must be the case that
$\eta(\freq \sigma)_i = 0$.

For the final part of the lemma, note that one of the conditions holds iff
\[
\max_t |(\freq \sigma_2^{(t)})_i - (\freq \sigma_2)_i| \ge M
\]
or equivalently iff \[
\max_{t \ge |\sigma_1|}|(\freq \sigma^{(t)})_i - (\freq \sigma)_i| \ge M
\]
which as $(\freq \sigma)_i = \pm M$, holds iff there is a $t \ge |\sigma_1|$
such that either $(\freq \sigma^{(t)})_i \le 0$ and $(\freq \sigma)_i = M$, or
$(\freq \sigma^{(t)})_i \ge 0$ and $(\freq \sigma)_i = -M$, and in turn one of these
holds iff $\zeta(\freq \sigma^{(t)}) \not= \zeta(\freq \sigma)$.
\end{proof}

The algorithm is described in Algorithm~\ref{alg:weak2m}.

\begin{algorithm}\caption{Low-probability $\promise$ over $\Gamma_{2M-1}$}\label{alg:weak2m}
\addtolength\linewidth{-4.6ex}

\begin{enumerate}
\item Let $(a, b, c)$ be a uniformly chosen random labeling of $\Vt$.  Choose
  $u \in V_a$ uniformly at random.
\item While passing through the stream:
  \begin{enumerate}
  \item Track all updates to $x^{ab,a,u}$ and $x^{ac, a, u}$.
  \item While doing so, keep checking whether $\zeta(x^{ab,a,u})$ is a valid
  inner encoding of $\Sigma$; if it is, and it doesn't decode to $\perp$, then
  it is an encoding of $(f_b(v'), z)$ for some $v' \in V_b$ and $z'$.  Let
  $(v', z')$ be those values, if they exist.
  \item As soon as $(v', z')$ is set, track all updates to $x^{bc,b,v'}$,
  recording the current, minimum, and maximum value of each of its
  coordinates.  Discard these updates whenever $(v', z)$ changes.
\end{enumerate}
\item After the stream finishes:
  \begin{enumerate}
  \item Decode $\zeta(x^{ab,a,u})$ and $\zeta(x^{ac, a, u})$ to $\Sigma$.
  \item If either is $\perp$, output $\perp$.
  \item Otherwise, let their decodings be $(f_b(v), z_{uv})$ and
    $(f_c(w), z_{uw})$ for $v \in V_b$ and $w \in V_c$.
  \item If the final observed value for $x^{bc,b,v}$ is within $M-1$ of all the
  values the algorithm has observed for it, output $\perp$.
  \item Otherwise, by Lemma \ref{lem:2mdecoding} it knows
  $\eta(x^{bc,b,v})_i$ for some index $i \in [B]$.
    Let $z_{vw} = \eta(x^{bc,b,v})_i \oplus f_c(w)^{(bin)}_i$.
  \item Output $z_{uv} \oplus z_{vw} \oplus z_{uw}$.
  \end{enumerate}
\end{enumerate}
\end{algorithm}

\begin{lemma}
The space complexity of Algorithm~\ref{alg:weak2m} is $O(\log n\log M)$ bits.
\end{lemma}
\begin{proof} The randomness in step 1 uses $\log (6N)$ bits.  After that, the
algorithm tracks three length-$B$ vectors with entries in $\lbrace -M,
M\rbrace$; the total space usage is $O(\log n\log M)$.
\end{proof}

\begin{lemma}\label{lem:weak2mcorrect}
Algorithm~\ref{alg:weak2m} outputs either $\perp$ or $\tau$.  If $u$ is part of
a triangle in the underlying $\promise(n)$ graph $G$, and the last time
$\zeta(x^{ab,a,u})$ differs from its final value is before the last time
$\zeta(x^{bc,b,v})$ differs from its final value, then the algorithm outputs
$\tau$.
\end{lemma}
\begin{proof}
Note that $x^{ab,a,u}$ and $x^{ac, a, u}$ are tracked completely, so their
final decodings into $\Sigma$ are correct.  If $u$ is not part of a triangle,
at most one edge is incident to $u$ in the full graph $G$, so at least one of
the decodings is $\perp$ and the algorithm returns $\perp$.

Otherwise, if $u$ is part of a triangle, the algorithm correctly deduces $(v,
z_{uv})$ and $(w, z_{uw})$. If the last time $\zeta(x^{ab,a,u})$ differs from
its final value is \emph{after} the last time $\zeta(x^{bc,b,v})$ differs
from its final value, then at the time the algorithm starts tracking
$x^{bc,b,v}$, $\zeta(x^{bc,b,v})$ has already its final value, and so by
Lemma~\ref{lem:2mdecoding}, the final observed value for $x^{bc,b,v}$ is
within $M -1$ of all the values observed for it, and so the algorithm outputs
$\perp$. Otherwise, by Lemma~\ref{lem:2mdecoding}, the algorithm correctly
identifies $\eta(x^{bc,b,v})_i$.

Since $uvw$ is a triangle, we know player
$bc$ has the input $(v, w, z_{vw})$ for some $vw$, and we know
\[
\eta(x^{bc,b,v})_i = z_{vw} \oplus f_c(w)^{(bin)}_i.
\]
Thus the algorithm correctly identifies $z_{vw}$, and the
$\promise(n)$ promise says
\[ 
\tau = z_{uv} \oplus z_{vw} \oplus z_{uw}.
\]
Hence the algorithm outputs either $\perp$ or $\tau$, and the last time
$\zeta(x^{ab,a,u})$ differs from its final value is before the last time
$\zeta(x^{bc,b,v})$ differs from its final value, then the algorithm outputs
$\tau$.
\end{proof}

\begin{lemma}\label{lem:weak2mchance}
Algorithm~\ref{alg:weak2m} outputs $\tau$ with at least
$\frac{1}{180}$ probability.
\end{lemma}
\begin{proof}
There is a $n/N = 1/30$ chance that $u$ lies in a triangle, independent of
the choice of $(a, b, c)$. Furthermore, if it does, \emph{which} triangle it
lies in is independent of the choice of $(a, b, c)$.

Suppose $u$ lies in the triangle $uvw$ with $u \in V_{a'}, v \in V_{b'}, w \in
V_{c'}$. WLOG, let $\zeta(x^{a'b', a', u})$ stop changing before
$\zeta(x^{b'c', b', v})$ or $\zeta(x^{c'a', c', w})$.

Then Lemma~\ref{lem:weak2mcorrect} says that if $(a, b, c) = (a', b', c')$,
Algorithm~\ref{alg:weak2m} will output $\tau$.  This choice happens with $1/6$
probability; combined with the $1/30$ chance that $u$ lies in a triangle, we
get at least a $1/180$ chance of outputting $\tau$.
\end{proof}

\begin{lemma}
There is a turnstile streaming algorithm that solves $P_n$ on $\Gamma_{2M - 1}$ 
with probability $2/3$ using $O(\log n \log M)$ bits of space.
\end{lemma}
\begin{proof}
Run Algorithm~\ref{alg:weak2m} in parallel $360$ times and output any
non-$\perp$ result.  By Lemma~\ref{lem:weak2mcorrect} any non-$\perp$ result
will be correct.  By Lemma~\ref{lem:weak2mchance} the failure probability is at
most $(1 - 1/180)^{360} < 1/e^2 < 1/3$.
\end{proof}

\section{Restricted Intermediate State Triangle Counting}
\subsection{Problem}
Valid inputs to our problem will be as follows (for invalid inputs, any output
is accepted): $x$ will be a binary string indexed by $E(K_n)$, the set of all
possible edges on an $n$-vertex graph. We will associate it with a graph $G$ on
$n$ vertices with edge set $\lbrace e \in E(K_n) : x_e = 1\rbrace$. We will use
$m$ to denote the size of this edge set. Finally, $G$ has max degree $d$.

Instead of bounding the length of the stream, we will require that $x^{(t)}$ 
correspond to a graph $G$ with max degree $d$ for \emph{all} $t$. One
consequence of this is that all updates will be in $\lbrack -1, 1\rbrack$.

Our problem will be to estimate $T$, the number of triangles in the graph, up
to some multiplicative precision $\varepsilon$. Our algorithm will succeed in
doing this if the space allocated to it is large enough in terms of $T$. This
space requirement is decreasing in $T$, so we may express this as a data stream
problem in the sense of Definition~\ref{dfn:streamproblem} by choosing a lower
bound $T'$ and making any answer acceptable for an input vector $x$ that does
not correspond to a valid input or results in $T < T'$, and making all outputs
in $\lbrack (1 - \varepsilon)T, (1 + \varepsilon)T\rbrack$ acceptable for input
vectors that correspond to a valid graph with $T \ge T'$.

\subsection{Linear Sketching Lower Bound}
By Theorem~7 of~\cite{KKP18}, any sketching algorithm for this problem requires
$\Omega(m/T^{1/3})$ bits. The requirement that $d$ be constant does not affect
this, as the \cite{KKP18} reduction is on graphs of max degree 2. Neither does
the intermediate state requirement, as the output of a sketching algorithm
depends only on the final state of the stream.

\subsection{Algorithm}
\begin{enumerate}
\item Initialize our set of seed edges $S = \emptyset$. Let $h : E \rightarrow
\lbrace 0, 1\rbrace$ be a threewise independent hash function where $h(e) = 1$
with probability $p$.
\item While passing through the stream:
\begin{enumerate}
\item On receiving an update $(e, +1)$:
\begin{itemize}
\item If $h(e) = 1$ and $|S| \le 2pm$, add $e$ to $S$, and initialize $S_e$ as
$\emptyset$.
\item If $\exists f \in S$ such that $e$ is incident to $f$, add $e$ to $S_f$.
\end{itemize}
\item On receiving an update $(e, -1)$:
\begin{itemize}
\item Remove it from any of $S$ and the sets $S_f$ that contain it. 
\item Delete the set $S_e$ if it exists.
\end{itemize}
\end{enumerate}
\item For each $e = uv$, set \[
\wt{T}_e = \begin{cases}
p^{-1}|\lbrace w : uw,vw \in S_e\rbrace| & \mbox{if $e \in S$}\\
0 &\mbox{otherwise.}
\end{cases}
\]
\item Return $\wt{T} = \sum_e \wt{T}_e$.
\end{enumerate}

\subsection{Space Complexity}
\begin{lemma}
This algorithm requires $O(pdm \log n)$ bits of space.
\end{lemma}
\begin{proof}
The set $S$ has size at most $2pm$ at any point in time, and for each element
$e$ in $S$ at most $2d - 1$ edges are kept (as each endpoint of $e$ has degree
at most $d$ at all times), and each edge takes $O(\log n)$ bits of space to store.
\end{proof}

\subsection{Correctness}
\begin{definition}
$G^{(t)}$ and $S^{(t)}$ denote the state of $G$ and $S$
respectively after the first $t$ updates, so that $G^{(L)} = G$ and $S^{(L)} =
S$.
\end{definition}

\begin{definition}
For any edge $e \in G$, let $t_e$ denote the time of the last update made to
$e$. For any triangle $\tau \in G$, let $\rho(\tau)$ denote the edge $e \in
\tau$ that minimizes $t_e$. Then: \[
T_e = |\lbrace \tau : \rho(\tau) = e\rbrace|
\]
\end{definition}

Note that as each triangle $\tau$ has exactly one $e$ such that $\rho(\tau) =
e$, $T = \sum_e T_e$.

\begin{definition}
Let $Q^{(t)} =  \lbrace e \in E(G^{(t)}) : h(e) = 1\rbrace$, 
$Q = Q^{(L)}$, and $Q_e = \lbrace f\text{ incident to }e : t_f > t_e\rbrace$. Then:
\begin{align*}
\wt{T}^{+}_e &= \begin{cases}
p^{-1}|\lbrace w : uw,vw \in Q_e\rbrace|  &\mbox{if $e \in Q$}\\
0 &\mbox{otherwise.}
\end{cases}\\
\wt{T}^{+} &= \sum_e \wt{T}^{+}_e
\end{align*}
\end{definition}
\begin{lemma}
\label{lm:mtdist}
\begin{align*}
\E{\wt{T}^{+}} &= T\\
\var(\wt{T}^{+}) &\le p^{-1}dT
\end{align*}
\end{lemma}
\begin{proof}
For each $e \in E(G)$, $\wt{T}^{+}_e = p^{-1}T_e$ if $h(e) = 1$ and $0$
otherwise. So
\begin{align*}
\E{\wt{T}^{+}_e} &= T_e\\
\var(\wt{T}^{+}_e) &\le T_e^2/p\\
&\le dT_e/p
\end{align*}
and as $h$ is threewise independent:
\begin{align*}
\E{\wt{T}^{+}} &= \sum_e T_e\\
&= T\\
\var(\wt{T}^{+})&= \sum_e \var(\wt{T}^{+}_e)\\
&\le dT/p\text{.}
\end{align*}
\end{proof}

\begin{lemma}
\label{lm:qeqs}
For any $e \in Q$, if $|S^{(t_e - 1)}| < 2pm$, $\wt{T}_e = \wt{T}^{+}_e$.
Otherwise, $\wt{T}_e = 0$.
\end{lemma}
\begin{proof}
If $e \in Q$, it will be in $S$ unless $S$ is size $2pm$ at the final time it
would be added (if it is added earlier, it will be deleted before time $t_e$,
so only the size of $S^{(t_e)}$ matters). Furthermore, if it is added, the
edges in $S_e$ will be precisely those edges of $G$ that have their final
update after $S_e$ is created for the last time, that is, after $t_e$. So
if $|S^{(t_e - 1)}| < 2pm$, $\wt{T}_e = \wt{T}^{+}_e$.

On the other hand, if  $|S^{(t_e - 1)}| = 2pm$, then $e \not\in S^{(t_e - 1)}$,
as it will have been deleted since the last time it might have been added,
$e \not\in S^{(t_e)}$, as it will not be added, and so $e \not\in S$, as there
are no more updates to $e$.
\end{proof}

\begin{lemma}
\label{lm:sbigprob}
For all $e \in E(G)$: \[
\Pb{|S^{(t_e - 1)}| = 2pm \middle| h(e) = 1} \le 1/pm
\]
\end{lemma}
\begin{proof}
By the intermediate state condition on $G^{(t_e - 1)}$, it has at most $m$ edges.
Then as $S^{(t_e - 1)} \subseteq Q^{(t_e - 1)}$, and as $h$ is threewise independent and
$h(e) = 1$ with probability $p$,
\begin{align*}
\E{|Q^{(t_e - 1)}| \middle| h(e) = 1} &\le pm\\
\var\left(|Q^{(t_e - 1)}| \middle| h(e) = 1\right) &\le (p - p^2)m
\end{align*}
\noindent
so by Chebyshev's inequality: 
\begin{align*}
\Pb{|S^{(t_e - 1)}| = 2pm\middle|h(e) = 1} &\le \Pb{|Q^{(t_e - 1)}| \ge 2pm\middle|
h(e) = 1}|\\
&\le 1/pm
\end{align*}
\end{proof}
\begin{lemma}
\label{lm:wtctdif}
\[
\E{|\wt{T} - \wt{T}^{+}|} \le T/pm
\]
\end{lemma}
\begin{proof}
By Lemma \ref{lm:qeqs}, $|\wt{T}_e - \wt{T}^{+}_e| = p^{-1}T_e$ if $h(e) = 1$
and $|S^{(t_e - 1)}| = 2pm$, and $0$ otherwise. So, Lemma \ref{lm:sbigprob}:
\begin{align*}
\E{|\wt{T} - \wt{T}^{+}|} &\le \sum_e \E{|\wt{T}_e - \wt{T}^{+}_e|}\\
&\le \sum_e p^{-1}T_e \Pb{ |S^{(t_e - 1)}| = 2pm \wedge h(e) = 1}\\
&\le \sum_e (T_e/p^2m)\Pb{h(e) = 1}\\
&= T/pm
\end{align*}
\end{proof}

\restate{thm:maxd}
\begin{proof}
Let the algorithm be run with $p = 32d/\varepsilon^2 T$.
Then by Lemma \ref{lm:wtctdif}, 
\begin{align*}
\E{|\wt{T} - \wt{T}^{+}|} &\le T^2/32dm\\
&\le T/32 &\mbox{as $T \le dm$.}
\end{align*}
\noindent
Therefore, by Markov's inequality: \[
\Pb{|\wt{T} - \wt{T}^{+}| \ge \varepsilon T/2} \le 1/16
\]
Then, by Lemma \ref{lm:mtdist}, 
\begin{align*}
\E{\wt{T}^{+}} &= T\\
\var(\wt{T}^{+}) &\le T^2/8
\end{align*}
\noindent
and so by Chebyshev's inequality, \[
\Pb{|\wt{T}^{+} - T| \ge \varepsilon T/2} \le 1/4
\]
so: \[
\Pb{|\wt{T} - T| \ge \varepsilon T} \le 5/16
\]
Therefore, by running $O(\log 1/\delta)$ copies of the algorithm in parallel
and taking the median, we can output a $(1 \pm \varepsilon)$ multiplicative
approximation to $T$ with probability $1 - \delta$.
\end{proof}
\section{Bounded-Length Triangle Counting}
\subsection{Problem}
We will work in the \emph{strict} turnstile model, so our input vector $x =
\freq \sigma^{(L)}$ is non-negative at all intermediate steps.

Valid inputs to our problem will be as follows (for invalid inputs, any output
is accepted): $x$ will be indexed by $E(K_n)$, the set of all possible edges on
an $n$-vertex graph. We will associate it with a graph $G$ on $n$ vertices with
edge set $\lbrace e \in E(K_n) : x_e = 1\rbrace$.  $x$ is binary, but its
intermediate states may not be.  We will use $m$ to denote the size of this
edge set. Finally, $G$ has max degree $d$.

Our problem will be to estimate $T$, the number of triangles in the graph, up
to some multiplicative precision $\varepsilon$. Our algorithm will succeed in
doing this if the space allocated to it is large enough in terms of $T$. This
space requirement is decreasing in $T$, so we may express this as a data stream
problem in the sense of Definition~\ref{dfn:streamproblem} by choosing a lower
bound $T'$ and making any answer acceptable for an input vector $x$ that does
not correspond to a valid input or results in $T < T'$, and making all outputs
in $\lbrack (1 - \varepsilon)T, (1 + \varepsilon)T\rbrack$ acceptable for input
vectors that correspond to a valid graph with $T \ge T'$.
 
\subsection{Linear Sketching Lower Bound}
By Theorem~7 of~\cite{KKP18}, any sketching algorithm for this problem requires
$\Omega(m/T^{1/3})$ bits. The requirement that $d$ be constant does not affect
this, as the \cite{KKP18} reduction is on graphs of max degree 2, and neither
do the stream length and strict turnstile requirements, as they will not affect
the output of any linear sketch.

\subsection{Algorithm}
\begin{enumerate}
\item Initialize our set of seed edges $S = \emptyset$. Let $h : E \rightarrow
\lbrace 0, 1\rbrace$ be a pairwise independent hash function where $h(e) = 1$
with probability $p$.
\item While passing through the stream, on receiving an update $(e, \chi)$:
\begin{itemize}
\item If $h(e) = 1$, and there is no tuple $(e,\gamma) \in S$, add $(e, \chi)$ to $S$.
\item If $h(e) = 1$, and $(e,\gamma) \in S$, replace it with $(e,\chi + \gamma)$.
\item If $(e,\chi)$ has been added to $S$ for some $\chi > 0$, initialize the set
$S_e = \emptyset$.
\item If $(e,0)$ is now in $S$, delete $S_e$.
\item Then, for each $f$ incident to $e$ such that $(f,z) \in S$:
\begin{itemize}
\item If $(e,\gamma) \in S_f$, replace it with $(e, \max(\chi + \gamma,0))$.
\item Otherwise, insert $(e,\max(\chi,0))$ into $S_f$, unless $|S_f| \ge
\frac{2d^2L}{\varepsilon T}$.
\end{itemize}
\end{itemize}
\item For each edge $e = uv$, set: \[
\wt{T}_e = \begin{cases}
p^{-1}|\lbrace w :  (uw, 1), (vw, 1) \in S_e\rbrace| & \mbox{If
$(e,1) \in S$.}\\
0 & \mbox{Otherwise.} \end{cases}
\]
\item Return $\wt{T} = \sum_e \wt{T}_e$.
\end{enumerate}

\subsection{Space Complexity}
\begin{lemma}
The expected space complexity of this algorithm is at most
$O\left(\frac{pd^2L^2}{\varepsilon T}  \log n\right)$ bits.
\end{lemma}
\begin{proof}
Each edge in the stream is independently included in $S$ with probability $p$,
so the expected maximum size of $S$ is at most $pL$. For each element of $S$ we
keep an integer of size $\text{poly}(n)$, requiring $O(\log n)$ bits, and a set
of size no more than $\frac{2d^2L}{\varepsilon T}$. The elements of these sets
are edges of an $n$-vertex graph, and integers of size $\text{poly}(n)$, and
therefore require $O(\log n)$ bits each to represent.
\end{proof}

\subsection{Correctness}
Consider some fixed (strict) turnstile stream of length $L$. Let $G$ be the
graph with vertex set $\lbrack n\rbrack$ and edge set $\lbrace e \in E : x_e =
1 \rbrace$, and let $T$ be the number of triangles in $G$. We will seek to show
that this algorithm can approximate $T$.
\begin{definition}
For any edge $e \in G$, let $t_e$ be the largest $t \in \lbrack L\rbrack$
such that: \begin{align*}
x_e^{(t-1)} &= 0\\ 
x_e^{(t)} &> 0
\end{align*}
For any triangle $\tau \in G$, let $\rho(\tau) \in \tau$ be the edge of $\tau$
that maximizes $t_{\rho(\tau)}$. Then, define: \[
T_e = |\lbrace \tau \in G : \rho(\tau) = e\rbrace|
\]
\end{definition}
\noindent
Note that as each triangle $\tau$ has exactly one edge $e$ such that
$\rho(\tau) = e$, $\sum_e T_e = T$.
\begin{definition}
For any edge $e \in G$ and $t \ge t_e$, $Q_e^{(t)}$ is the set generated by the
following procedure: 
\begin{itemize}
\item For $t' = t_e,\dots,t$, and $(f,\chi) = \sigma_{t'}$, if $f$ is incident to $e$:
\begin{itemize}
\item If $(f,\gamma) \in Q_e^{(t)}$, replace it with $(e,\max(\chi + \gamma, 0))$.
\item Otherwise, insert $(f,\max(\chi,0))$ into $Q_e^{(t)}$.
\end{itemize}
\end{itemize}
\end{definition}
\begin{lemma}
\label{lm:qsupset}
For any $e$ such that $h(e) = 1$,
\[
Q_e^{(L)} \supseteq S_e
\]
with equality when \[
|Q_e^{(L)}| \le \frac{2d^2L}{\varepsilon}\text{.}
\]
\end{lemma}
\begin{proof}
As $h(e) = 1$, $S_e$ will be deleted and recreated for the final time at $t_e$.
After this point, the procedures for creating $S_e$ and $Q_e^{(L)}$ are
identical as long as $|S_e|$ (and therefore $Q_e^{(L)}$) never reaches size
$\frac{2d^2L}{\varepsilon}$. If it does, the only difference is that some edges
may be excluded from $S_e$.
\end{proof}
\noindent
For any $(f,z)$ such that $f \in Q_e^{(t)}$ we will also write $f \in
Q_e^{(t)}$, and $Q_e^{(t)}\lbrack f\rbrack = z$. Note that $Q_e^{(r)}\lbrack
f\rbrack$ is well-defined whenever $f \in Q_e^{(t)}$ (as no edge is added to
$Q_e^{(t)}$ more than once) and $f \in Q_e^{(t)} \Rightarrow f \in
Q_e^{(t+1)}$ (as no edges are ever removed from $Q_e^{(r)}$.

\begin{lemma}
\label{lm:qeval}
For all edges $f$ incident to $e$ and integers $t \in \lbrack t_e, L\rbrack$,
\[
Q_e^{(t)}\lbrack f\rbrack = x_f^{(t)} - \min_{r = t_e, \dots, t} x_f^{(r)}
\]
\end{lemma}
\begin{proof}
We proceed by induction on $t$. If $t = t_e$, as the update at time $t_e$ was
to $e$, $Q_e^{(t)}\lbrack f\rbrack = 0$ and so the result holds. Now suppose $t
> t_e$ and $Q_e^{(t-1)}\lbrack f\rbrack = x_f^{(t-1)} - \min_{r = t_e, \dots,
t-1} x_f^{(r)}$.

Then, let $\sigma_t = (f',\chi)$. If $f' \not= f$ both sides of the equation
are unchanged and we are done. So suppose the update is $(t,f,\chi)$.  We will
consider two cases.
\begin{description}
\item[$\mathbf{Q_e^{(t-1)}\lbrack f\rbrack + \chi \ge 0}$] Then $Q_e^{(t)} =
Q_e^{(t-1)}\lbrack f\rbrack + \chi$ and $x_f^{(t)} = x_f^{(t - 1)} + \chi$.
Furthermore, $\chi \ge - Q_e^{(t-1)}\lbrack f\rbrack$, so we have:
\begin{align*}
x_f^{(t)} &= x_f^{(t-1)} + \chi\\
&\ge x_f^{(t-1)} - Q_e^{(t-1)}\lbrack f\rbrack\\
&=  \min_{r = t_e, \dots, t-1} x_f^{(r)}
\end{align*}
\noindent
So $\min_{r = t_e, \dots, t} x_f^{(r)} = \min_{r = t_e, \dots, t-1} x_f^{(r)}$,
completing the proof.

\item [$\mathbf{Q_e^{(t-1)}\lbrack f\rbrack + \chi < 0}$] Then $Q_e^{(t)} = 0$, and:
\begin{align*}
x_f^{(t)} &= x_f^{(t-1)} + \chi\\
&< x_f^{(t-1)} - Q_e^{(t-1)}\\
&= \min_{r = t_e, \dots, t-1} x_f^{(r)}
\end{align*}
\noindent
So $\min_{r = t_e, \dots, t} x_f^{(r)} = x_f^{(t)}$, and so $x_f^{(t)} - \min_{r = t_e, \dots, t} x_f^{(r)} = 0$, completing the proof.
\end{description}
\end{proof}

\begin{definition}
For any vertex $x$, let the `stream degree' $l_v$ be the number of edges $e$
incident to $x$ such that there is some update $\sigma_t = (e, \chi)$,
regardless of whether $e$ is in the final graph $G$.
\end{definition}

\begin{lemma}
\label{lm:wttedist}
Let $e = uv$ be an edge. Then
\[
\wt{T}_e = \begin{cases}
\overline{T}_e/p & \mbox{with probability $p$}\\
0 & \mbox{otherwise.}
\end{cases}
\]
where $\overline{T}_e = T_e$ if $l_u + l_v \le \frac{2d^2L}{\varepsilon T}$,
and $\overline{T}_e \in \lbrack 0,  T_e\rbrack$ otherwise.
\end{lemma}
\begin{proof}
Let $e$ be an edge. If $h(e) = 0$, $(e, 1) \not\in S$, and so $\wt{T}_e = 0$.
This event happens with probability $1 - p$. If $h(e) = 1$ but $e \not\in G$,
$x_e = 0$, and so $(e, 1) \not\in S$, so $\wt{T}_e = 0
= T_e = \wt{T}_e$. 

Now consider the case where $h(e) = 1$ and $e$ in $G$. Then $x_e = 1$, so
$(e,1) \in S$ . $\wt{T}_e$ will then be $p^{-1}$ times the number of triangles
$uvw$, where $e = uv$ and $(uw,1),(vw,1) \in S_e$. If $l_u + l_v \le
\frac{2d^2L}{\varepsilon T}$,  then $|Q_e^{(L)}| \le \frac{2d^2L}{\varepsilon
T}$ and so by Lemma \ref{lm:qsupset}, $Q_e^{(L)} = S_e$, and otherwise
$Q_e^{(L)} \supseteq S_e$.

So it will suffice to show that \[
|\lbrace w :  (uw, 1), (vw, 1) \in Q_e^{(L)}\rbrace| =  |\lbrace \tau \in G :
\rho(\tau) = e\rbrace|
\]
. We will show that \[
\lbrace f : (f, 1) \in Q_e^{(L)}\rbrace = \lbrace f \in G : t_f > t_e 
\text{, $e$ incident to $f$} \rbrace
\]
which implies our result, as it means that $w \in \lbrace w :  (uw, 1), (vw, 1)
\in Q_e^{(L)}\rbrace$ iff the triangle $uvw$ has $t_{uv} < t_{uw},t_{vw}$.

For any $f \in E$ incident to $e$, by Lemma \ref{lm:qeval}, $(f,1) \in
Q_e^{(L)}$ iff $x_f^{(L)} - \min_{r = t_e, \dots, L} x_f^{(r)} = 1$. If $f
\not\in G$, then $x_f^{(L)} = 0$ and so this cannot hold, as $x_f^{(r)} \ge 0$
for all $r$. If $f \in G$, then $x_f^{(L)} = 1$ and so this holds iff $\min_{r
= t_e, \dots, L} x_f^{(r)} = 0$, that is, iff $t_f > t_e$. So $(f,1) \in
Q_e^{(L)}$ iff $f \in G$ and $t_f > t_e$, concluding the proof.
\end{proof}
\begin{lemma}
\label{lm:bdexp}
\[
\E{\wt{T}} \in \lbrack (1 - \varepsilon/2)T, T\rbrack
\]
\end{lemma}
\begin{proof}
By Lemma \ref{lm:wttedist}, $\E{\wt{T}} = \sum_e\overline{T}_e$, where
$\overline{T}_e = T_e$ if $l_u + l_v \le \frac{2d^2L}{\varepsilon T}$ and
$\overline{T}_e \in \lbrack 0, T_e\rbrack$ otherwise. Recalling that $T_e =
|\lbrace \tau \in G : \rho(\tau) = e\rbrace|$, this gives us \[
\E{\wt{T}} \le T
\]
and \[
\E{\wt{T}} \ge \sum_{\substack{uv:\\l_u + l_v \le \frac{2d^2L}{\varepsilon
T}}}T_{uv}\text{.}
\]
The right-hand side of the second expression is precisely the number of
triangles $\tau$ in $G$ such that $\rho(\tau) = uv$ with $l_u + l_u \le
\frac{2d^2L}{\varepsilon T}$. So let $T^-$ be the number of triangles that do
\emph{not} satisfy this criterion. For each such triangle $\tau$, there are at
least $l_u + l_v$ updates in $\Sigma$ to edges incident to $\rho(\tau)$.
Furthermore, as the final graph has max degree $d$, at most ${d \choose 2} \le
d^2/2$ triangles use any vertex. So we have: 
\begin{align*}
L &\ge \frac{1}{2}\sum_v l_v\\
&\ge \frac{1}{d^2}\sum_{\substack{\tau, uv :\\ \rho(\tau) = uv}} l_u + l_v\\
&\ge \frac{1}{d^2}T^-\frac{2d^2L}{\varepsilon T}
\end{align*}
So $T^- \le \varepsilon T /2$, and the result follows.
\end{proof}

\begin{lemma}
\label{lm:bdvar}
\[
\var(\wt{T}) \le p^{-1}dT
\]
\end{lemma}
\begin{proof}
For any fixed stream $\Sigma$, each $\wt{T}_e$ depends only on whether $h(e) =
1$, and so as $h$ is pairwise independent, so are the $\wt{T}_e$, and so:
\begin{align*}
\var(\wt{T}) &= \sum_e \var(\wt{T}_e)\\
&\le \sum_e \E{\wt{T}_e^2}\\
&\le \sum_e \Pb{h(e) = 1}p^{-2}T_e^2\\
&\le \sum_e p^{-1}dT_e\\
&= p^{-1}dT
\end{align*}
\end{proof}

\restate{thm:boundedl}
\begin{proof}
By Lemma \ref{lm:bdvar}, we may set $p$ in the above algorithm to be
$\frac{16d}{\varepsilon^2 T}$, so that the algorithm requires
$O\left(\frac{d^3L^2}{\varepsilon^2T^2}\log n\right)$ space and $\var(\wt{T}) =
\frac{\varepsilon^2 T^2}{16}$. Then, by Chebyshev's inequality, the probability
that $|\wt{T} - \E{\wt{T}}| \ge \varepsilon T/2$ is at most $1/4$. 

We may then repeat the algorithm $O(\log 1/\delta)$ times in parallel, taking
the median, so that our final output is within $\varepsilon T/2$ of
$\E{\wt{T}}$ with probability $1 - \delta$. By Lemma \ref{lm:bdexp}, this
implies it is within $\varepsilon T$ of $T$.
\end{proof}

\section{Deterministic Turnstile-Sketching Equivalence}\label{sec:equivalence}
\subsection{Overview}
We will show that deterministic turnstile streaming algorithms can be expressed
as linear sketches. Here these sketches will take the form of linear functions
$\phi$ from $\Z^n$ to a module $M$ whose elements can be stored in $s$ space,
where $s$ is the space used by the turnstile streaming algorithm $\alg$. 

$M$ and $\phi$ will be characterized by ``moduli'' $a_i$ and ``overflow
vectors'' $o_i$ supported on indices smaller than $i$. A vector $x$ in $M$ is
simply a vector in $\prod_{i=1}^n \Z_{a_i}$, but instead of addition being
coordinatewise mod $(a_i)_n$, a coordinate $i$ which becomes larger than
$a_i$ will ``overflow'', with $o_i$ added to $x$ for every time $a_ie_i$ has to
be subtracted. This can cause repeated overflows, but as $o_i$ is only
supported on indices smaller than $i$, eventually these will stop.

By the structure theorem for $\Z$-modules, $M$ is isomorphic to some
direct product of cyclic modules, but this isomorphism is not (to our
knowledge) necessarily calculable in small space. However, because our
sketch $\phi$ represents a module, it has all the desirable properties
of linear sketches: it is mergeable, automatically allows deletions,
and is indifferent to stream length and order.

We will start by defining $M$ in terms of the parameters $a_i$ and
$o_i$, showing that if the parameters can be calculated in small space
then the homomorphism can also be calculated in small space.  We will
then give two methods of generating these parameters, and show that
the corresponding sketches can be used to solve stream problems,
proving equivalence first for total functions:

 \restate{thm:detreductiontotal}

 Then, for algorithms that can tolerate very long stream lengths, we prove
 equivalence for general stream problems:

 \restate{thm:detreductionpartial}

\subsection{Our Module}
\subsubsection{Definition of $M$}
Let $(a_i)_{i=1}^n$ be positive integers, and let at most $m$ of them be
greater than $1$. Let $(o_i)_{i=1}^n$ be vectors such that for all $i$, $o_i
\in \prod_{j=1}^{i-1}\Z_{a_j} \times \lbrace 0\rbrace^{n - i + 1}$. We will define
\[
  M = \left(\prod_{i=1}^n \Z_{a_i}, \star\right)\\
\]
a $\Z$-module with $\star$ as its addition operation.  We will now recursively
define a homomorphism $\phi : \Z^n \rightarrow M$, and then use this to define
$\star$.  
\begin{itemize}
  \item $\phi(\mathbf{0}) = \mathbf{0}$
  \item For $i \in \lbrack n\rbrack$, and any vector $x + re_i$ where $x_j = 0$ for all $j \ge i$, $\phi(x + re_i) = (r \bmod a_i)e_i + \phi(x + (\lfloor r / a_i \rfloor)o_i)$.
\end{itemize}
This is well-defined because $x + (\lfloor r / a_i \rfloor)o_i$ is zero on all coordinates greater than $i - 1$.

We can now define $\star$ in terms of $\phi$, using the fact that every vector in $M$ is also a vector in $\Z^n$: \[
  x \star y = \phi(x + y)
\]

\subsubsection{Algebraic Properties of $M$ and $\phi$}
In this section we will prove that $M$ is in fact a $\Z$-module, and $\phi$ is
a homomorphism from $\Z^n$ to it.

\begin{lemma}
$\phi$ is idempotent.
\end{lemma}
\begin{proof}
As for any vector $x$ in $\Z^n$ that is also in $M$, $\phi(x) = x$.
\end{proof}
\begin{lemma}
  $\star$ is commutative.
\end{lemma}
\begin{proof}
  By the symmetry of the definition.
\end{proof}
\begin{lemma}
  $\star$ is associative.
\end{lemma}
\begin{proof}
  We need to prove that for any $x, y, z$, $(x \star y) \star z = x \star (y
  \star z)$. As $x \star y = \phi(x + y)$ and we have already shown that $\star$
  is commutative, it will suffice to prove that for all $x, y, z$, $\phi(\phi(x +
  y) +z) = \phi(x + y + z)$. We will prove this by induction on $i$, the smallest 
  non-negative integer such that for all $j > i$, $x_j = y_j = z_j = 0$.

  If $i = 0$, $x = y = z = \mathbf{0}$ and so as $\phi(\mathbf{0}) = \mathbf{0}$
  the result follows immediately. Otherwise, suppose the result holds for $i - 1$
  and let $x, y, z$ be such that for all $j > i$, $x_j = y_j = z_j = 0$. Then
  we may write
  \begin{align*}
    x &= x' + r_1e_i\\
    y &= y' + r_2e_i\\
    z &= z' + r_3e_i
  \end{align*}
  where $x'_j, y'_j, z'_j$ are zero for all $j > i - 1$. Then, by the inductive
  hypothesis,
  \begin{align*}
    \phi(\phi(x + y) + z) &= \phi(\phi(x' + r_1e_i + y' + r_2e_i) + z' + r_3e_i)\\
    &= \phi(\phi(x' + y' + \left\lfloor \frac{r_1 + r_2}{a_i}\right\rfloor o_i) + (r_1 + r_2 \bmod
    a_i)e_i  + z' + r_3e_i)\\
    &= \phi(\phi(x' + y' + \left\lfloor \frac{r_1 + r_2}{a_i}\right\rfloor o_i) +
    z' + \left\lfloor\frac{(r_1 + r_2 \bmod a_i) + r_3}{a_i}\right\rfloor o_i) \\
    &+ (r_1 + r_2 + r_3 \bmod
    a_i)e_i\\
    &= \phi(\phi(x' + y' + \left\lfloor \frac{r_1 + r_2}{a_i}\right\rfloor o_i +
    z' + \left\lfloor\frac{(r_1 + r_2 \bmod a_i) + r_3}{a_i}\right\rfloor o_i))\\
    &+ (r_1 + r_2 + r_3 \bmod
    a_i)e_i\\
    &= \phi(x' + y' + \left\lfloor \frac{r_1 + r_2}{a_i}\right\rfloor o_i +
    z' + \left\lfloor\frac{(r_1 + r_2 \bmod a_i) + r_3}{a_i}\right\rfloor o_i) + (r_1 + r_2 + r_3 \bmod
    a_i)e_i\\
    &= \phi(x' + y' +  z' + \left\lfloor\frac{r_1 + r_2 + r_3}{a_i}\right\rfloor
    o_i) + (r_1 + r_2 + r_3 \bmod a_i)e_i\\
    &= \phi(x' + y' +  z' + (r_1 + r_2 + r_3)e_i)\\
    &= \phi(x + y + z)
  \end{align*}
  as \[
    x' + y' + \left\lfloor \frac{r_1 + r_2}{a_i}\right\rfloor o_i +
  z' + \left\lfloor\frac{(r_1 + r_2 \bmod a_i) + r_3}{a_i}\right\rfloor o_i
  \]
has zeros at every coordinate greater than $i - 1$.
\end{proof}

\begin{lemma}
  $\forall x, y \in \Z^n, \phi(x + y) = \phi(x) \star \phi(y)$
\end{lemma}
\begin{proof}
  We proceed by induction on $i$, the smallest non-negative integer such that
  $x_j = y_j = 0$ for all $j > i$. If $i = 0$, then $x = y = \mathbf{0}$ and so
  the result follows immediately. So suppose $i > 0$ and the result holds for all
  smaller $i$. Let $x = x' + r_1e_i$, $y = y' + r_2e_i$, where $x'_j = y'_j = 0$
  for all $j \ge i$.
  \begin{align*}
    \phi(x + y) &= \phi(x' + y' + (r_1 + r_2)e_i)\\
    &= \phi(x' + y' + \left\lfloor\frac{r_1 + r_2}{a_i}\right\rfloor o_i) + (r_1 + r_2 \bmod a_i)e_i
  \end{align*}
  On the other hand:
  \begin{align*}
    \phi(x) \star \phi(y) &= \phi(\phi(x) + \phi(y))\\
    &= \phi(\phi(x' + \lfloor r_1/a_i\rfloor o_i) + \phi(y' + \lfloor r_2/a_i
    \rfloor o_i) + ((r_1 \bmod a_i) + (r_2 \bmod a_i))e_i)\\
    &= \phi(\phi(x' + \lfloor r_1/a_i\rfloor o_i) + \phi(y' + \lfloor r_2/a_i
    \rfloor o_i) + \left\lfloor\frac{(r_1 \bmod a_i) + (r_2 \bmod
    a_i)}{a_i}\right\rfloor o_i) \\
    &+  ((r_1 \bmod a_i) + (r_2 \bmod a_i) \bmod a_i)e_i\\
    &= \phi(x' + \lfloor r_1/a_i\rfloor o_i) + \phi(y' + \lfloor r_2/a_i
    \rfloor o_i) + \phi(\left\lfloor\frac{(r_1 \bmod a_i) + (r_2 \bmod
    a_i)}{a_i}\right\rfloor o_i)\\
    &+ (r_1 + r_2 \bmod a_i)e_i\\
    &= \phi(x' + y' + (\lfloor r_1/a_i\rfloor +  \lfloor r_2/a_i
    \rfloor + \left\lfloor\frac{(r_1 \bmod a_i) + (r_2 \bmod
    a_i)}{a_i}\right\rfloor)o_i) \\
    &+ (r_1 + r_2 \bmod a_i)e_i\\
    &=  \phi(x' + y' + \left\lfloor\frac{r_1 + r_2}{a_i}\right\rfloor o_i) + (r_1 + r_2 \bmod a_i)e_i\\
    &= \phi(x + y)
  \end{align*}

\end{proof}

\begin{lemma}
  $\star$ is invertible, with $\phi(-x)$ being the inverse of $\phi(x)$ for all
  $x \in \Z^n$. 
\end{lemma}
\begin{proof} 
  By the previous lemma,
  \begin{align*}
    \phi(x) \star \phi(-x) &= \phi(x + -x)\\
    &=\phi(\mathbf{0})\\
    &=\mathbf{0}
  \end{align*}
\end{proof}

Therefore, $M$ is an abelian group and so forms a $\Z$-module under the
natural definition of integer multiplication. 

\begin{lemma}
  \label{lem:homomorphism}
  $\phi : \Z^n \rightarrow M$ is a homomorphism of $\Z$-modules.
\end{lemma}
\begin{proof} 
  We already have that $\phi$ preserves addition and multiplication by $-1$, so
  it must also preserve multiplication by elements of $\mathbb{Z}$.
\end{proof}

\subsubsection{Space}
In this section, we will prove that, provided the moduli $a_i$ are small
enough and can be generated along with the $o_i$ in sufficiently small 
space, the sketch may be maintained in small space.
\define{thm:modulespace}{Theorem}{%
Suppose $\prod_{i=1}^na_i \le 2^s$, and for each $i$, $a_i$ and $o_i$ can be
calculated in $O(S + s + m\log n)$ space. Then the sketch $\phi(x)$ can be
stored in $s$ space and maintained under updates to $x$ using only $O(S + s +
m\log n + \log r)$ space for updates to $z$ of size $r$.
}
\state{thm:modulespace}
We now present an algorithm for calculating $\phi(x)$. All vectors are
stored as a list of indices and values.

\begin{algorithm}[H]
  \DontPrintSemicolon
  \caption{Calculating $\phi(x)$}\label{alg:phicalc}
  Calculate the moduli $a_i$.\;
  $z\gets x$ \;
  \While{$\exists i, z_i \ge a_i$}{
    Let $i$ be the smallest index such that $z_i \ge a_i$.\;
  Calculate $o_i$.\;
  $z_i \gets z_i - a_i$\;
  $z \gets z + o_i$\;
  Discard $o_i$.\;
}
  \bf{return} $z$
\end{algorithm}
\begin{lemma}
  Algorithm \ref{alg:phicalc} terminates.
\end{lemma}
\begin{proof}
  First note that, as $o_i$ is only supported on indices smaller than $i$, for
  any $j$, $z_j$ will not increase unless $i > j$, where $i$ is the smallest
  index such that $z_i \ge a_i$ (or $\infty$ if there is no such index). 

  Now, we we prove that for any $j$ from $0$ to $n$, and any starting value of
  $z$, it will only take a finite number of iterations of the inner loop of the
  algorithm until the first time $i > j$. We will prove this by a double
  induction on $j$ and $z_j$.

  Suppose $j = 0$. Then $i > j$ at the start of the stream. 

  Suppose $j > 0$ and $z_j < a_j$, and the result holds for all smaller values of
  $j$. By the inductive hypothesis, after some finite number of iterations we
  reach the first time that $i > j - 1$. As this is the first time, $z_j$ remains
  unchanged and so $i > j$. 

  Finally suppose $j > 0$, $z_j \ge a_j$, and the result holds for all $j, z_j$
  where at least one of $j$ and $z_j$ is smaller.  By the inductive hypothesis,
  after some finite number of iterations we reach the first time that $i > j -
  1$. At the next iteration, $z_j$ is reduced by $a_j$ and $o_j$ is added to $z$.
  By applying the inductive hypothesis with this new value of $z$, a finite
  number more steps will bring us to the first time that $i < j$.

  Therefore, by considering $j = n$ the algorithm will eventually terminate.
\end{proof}

\begin{lemma}
  When Algorithm \ref{alg:phicalc} terminates, it returns $\phi(x)$.
\end{lemma}
\begin{proof}
  At the end of the algorithm, the output is $z = \phi(z)$, as $\forall i, z_i <
  a_i$. At the start of the algorithm $z = x$ and so $\phi(z) = \phi(x)$. So it
  will suffice to show that each iteration of the algorithm leaves $\phi(z)$
  unchanged.

  An iteration picks some $i$ such that $z_i \ge a_i$ and replaces $z$ with $z +
  o_i - a_ie_i$. So we need to show that $\phi(z + o_i - a_ie_i)$. By Lemma
  \ref{lem:homomorphism}, $\phi$ is a homomorphism of $\mathbb{Z}$-modules.
  Therefore, 
  \begin{align*} 
    \phi(z + o_i - a_ie_i) &= \phi(z) \star \phi(a_ie_i - o_i)^{-1}\\ 
    &= \phi(z) \star( (a_i \bmod a_i)e_i + \phi(\lfloor a_i/a_i\rfloor
    o_i - o_i))^{-1}\\ 
    &= \phi(z) \star (\mathbf{0})^{-1}\\ 
    &= \phi(z) 
  \end{align*}
  concluding the proof.
\end{proof}

We now analyze the space complexity of updating this sketch. For the following
lemmas, we will assume that the conditions of Theorem~\ref{thm:modulespace}
hold. First we show that it is possible to store all the $a_i$ simultaneously.

\begin{lemma}
  The moduli $a_i$ can be stored in $O(s + m\log n)$ space.
\end{lemma}
\begin{proof}
  We can store the non-1 moduli as pairs $(i, a_i)$. The indices take $O(\log n)$
  bits to store, and the total space used by storing the values $a_i$ is at most
  $\sum_{i=1}^n \log a_i = \log \prod_{i=1}^n a_i \le s$.
\end{proof}

\begin{lemma}
  Algorithm \ref{alg:phicalc} uses $O(S + s + m\log n + \sum_{i\in
  \lbrack n\rbrack: x_i > 0}\log x_i + ||x||_0\log n)$ space.
\end{lemma}
\begin{proof}
  The space cost of the algorithm comes from calculating the moduli (which
  takes $O(S + s + m\log n)$ space), calculating $o_i$ (which takes $O(S + s +
  m\log n)$ space), storing $z$, and performing addition on coordinates of $z$
  (with the things to be added of size at most that of a coordinate of $o_i$
  or $a_i$, and therefore always smaller than the size of some modulus $a_j$
  for $j\le i$).

  Therefore, it will suffice to show that storing $z$ never requires more than
  $O(m\log n + \sum_{i\in \lbrack n\rbrack: x_i > 0}\log x_i + ||x||_0\log n)$
  space. First, note that a coordinate of $z$ only increases when $o_i$ is added
  to $z$, and this only happens when $z_j < a_j$ for every $j < i$. As each $o_i$
  is in $\prod_{j = 1}^n \Z_{a_j}$, this has two implications: 
  \begin{enumerate}
    \item At most $m + ||x||_0$ coordinates of $z$ are ever non-zero.
    \item Every non-zero coordinate $z_j$ is either no larger than $x_j$, or is at most twice $a_j$.
  \end{enumerate}
  The first of these two implies that we can store the indices $j$ such that $z_j
  > 0$ in at most $O((m + ||x||_0)\log n)$ space, while the second implies that
  we can store the list of values associated with these indices in at most
  $O(\sum_{i = 1}^n\log (2a_i) + \sum_{i\in \lbrack n\rbrack: x_i > 0}\log x_i) =
  O(s + \sum_{i\in \lbrack n\rbrack: x_i > 0}\log x_i)$ space.

\end{proof}
\begin{lemma}
  For any $x, y$ in $M$, $x \star y$ can be calculated in $O(S + s + m\log n)$
  space.
\end{lemma}
\begin{proof}
  $x \star y = \phi(x + y)$, so as $x$ and $y$ are both in $\prod_{i =1}^n
  \Z_{a_i}$, this follows directly from the previous lemma.
\end{proof}
We are now ready to prove that, for suitably generated $a_i$ and $o_i$, we may
maintain our sketch in small space.
\restate{thm:modulespace}
\begin{proof}
  We may store the sketch in only $s$ space by only storing the indices $i$
  where $a_i > 1$. We can then query it in $O(S + s + m\log n)$ space by
  calculating the moduli, and update it in space $O(S + s + m\log n + \log r)$
  for updates of size $r$ to $z$ by calculating $\phi(\phi(x) + re_i)$, where
  $i$ is the coordinate updated.
\end{proof}

\subsection{Sketching Total Functions}
\subsubsection{Overview}
In order to prove an equivalence between linear sketches and turnstile
algorithms for total functions, we need to define parameters $a_i$ and $o_i$ to
instantiate the linear sketch $\phi \rightarrow \Z^n$.

Once we have defined these parameters we will prove the sketch is ``correct''
--- for every $x \in \Z^n$, there is a stream with frequency $x$ on which
$\alg$ outputs the same thing as it does on $\kappa(\phi(x))$. We will then
show that it is possible to generate the parameters $a_i$ and $o_i$ in $O(s +
m\log n)$ space, and therefore by Theorem~\ref{thm:modulespace} we may maintain
the sketch in this space.

Finally, we will show that, using the streams described in the correctness
section, it is possible to recover a solution to any stream problem solved by
$\alg$ using the sketch.

\subsubsection{Defining the Parameters}
\label{sec:defineMtotal}
The $a_i$ and $o_i$ will be defined as the output of the following procedure,
which proceeds through the indices $i$ with backtracking.

For $i = 1, \dots, n$:
\begin{itemize}
  \item Let $x_j$ be defined as the $j^\text{th}$ vector in  $x \in \prod_{j =
    1}^{i - 1}\Z_{a_j} \times \Z \times \lbrace 0 \rbrace^{n - i}$ in
    little-endian order. 

    Let $j_2$ be the smallest integer such that there exists $j_1 < j_2$ such
    that $\alg(\kappa(x_{j_2})) = \alg(\kappa(x_{j_1}))$. Choose $a_i$, $o_i$
    so that $x_{j_2} - x_{j_1} = a_ie_i - o_i$. Note that $a_i \ge 0$ as
    $x_{j_2}$ is later than $x_{j_1}$ in little-endian order. If $a_i > 0$,
    move on to the next $i$.

    If $a_i = 0$, let $i'$ be the largest index such that $(x_{j_2} -
    x_{j_1})_{i'} > 0$. Choose $a_{i'}$ and $o_{i'}$ so that $x_{j_2} - x_{j_1}
    = a_{i'}e_{i'} - o_{i'}$, overwriting the old values of $a_{i'}$ and
    $o_{i'}$. Then roll $i$ back to $i' + 1$ and continue from there.
\end{itemize}

\begin{lemma}
  This procedure will terminate after a finite number of steps.
\end{lemma}
\begin{proof}
  After each iteration, either $i$ increases or $i$ is set to $i'+1$ with
  $a_{i'}$ reduced from its previous value. As the $a_i$ take values in the
  positive integers, the second can only happen finitely many times, and so the
  procedure will eventually terminate.
\end{proof}

\subsubsection{Space}
\label{sec:totalreductionspace}
\begin{lemma}
  $\prod_{i=1}^n a_i \le 2^s$.
\end{lemma}
\begin{proof}
  Consider the procedure from Section \ref{sec:defineMtotal}. In the final
  iteration (that is, when $a_n$ is defined rather than $i$ rolling back to
  some earlier index), $j_2$ was the smallest integer such that there existed
  $j_1$ such that $\alg(\kappa(x_{j_2})) = \alg(\kappa(x_{j_1}))$, and $a_n =
  (x_{j_2} - x_{j_1})_n$.

  As $j_2$ was the smallest integer such that this held, this implies that
  $\alg(\kappa(x_0),\dots,\alg(\kappa(x_{j_2-1}))$ were all distinct states. As
  the sequence $x_j$ comes from iterating through the vectors in
  $\prod_{i=1}^{n-1}\Z_{a_i} \times \Z$ in little-endian order, $j_2$ is at
  least $\prod_{i=1}^{n-1}a_i \times (x_{j_2})_n$. So as $a_n \le (x_{j_2})_n$,
  there are at least $\prod_{i=1}^n a_i$ distinct states of $\alg$, and so the
  result follows.
\end{proof}
\noindent
Recall that $m = |\lbrace i \in \lbrack n : a_i > 1\rbrace|$.
\begin{corollary}
\label{cor:msizetotal}
  $m \le s$
\end{corollary}
\begin{proof}
  This follows from the fact that the procedure that generates the $a_i$ will
  always roll back if it would generate an $a_i$ equal to 0, and therefore all
  the $a_i$ are positive integers.
\end{proof}

\begin{lemma}
  \label{lem:modulicalctotal}
  We may calculate all the moduli $a_i$ in $O(s + m\log n)$ space.
\end{lemma}
\begin{proof}
  To execute the procedure that generates the $a_i$, we need to remember the
  values of all $a_j$ for $j < i$ (which we can store in $O(s + m\log n)$ space,
  as at most $m$ are greater than $1$ and their magnitudes sum to at most $2^s$),
  and we need to find the pair $j_2 > j_1$ such that $\alg(\kappa(x_{j_2})) =
  \alg(\kappa(x_{j_1}))$. 

  We can generate any $\kappa(x_j)$ we will need in $O(S + s)$ space given a
  list of the $a_i$, as they just require marching through the elements of
  $\prod_{i=1}^{j-1} \Z_{a_i} \times \Z$ in little-endian order while executing
  the state-transition function of $\alg$, and the number of elements we go
  through is at most the number of distinct states of $\alg$.

  Therefore, we can find the pair in $O(S + s)$ space by running two copies of
  $\alg$ and feeding them the streams $\kappa(x_j)$ until we find a collision.
\end{proof}

\begin{lemma}
  \label{lem:overflowcalctotal}
  For any $i$, $o_i$ can be calculated in $O(S + s + m\log n)$ space.
\end{lemma}
\begin{proof}
  First note that, as each $o_i$ is in
  $\prod_{j=1}^{i-1}\Z_{a_j}\times\Z^{n-i+1}$, they can be \emph{stored} in $O(s
  + m\log n)$ space by storing $(j,(o_i)_j)$ pairs as above.

  To calculate $o_i$, we may first calculate all the $a_i$ as above, and then run
  the procedure until the final time where it changes $a_i$. At that point we may
  read off $o_i$ (as we know $x_j$ and $x_{j'}$).
\end{proof}

\subsubsection{Correctness}
\begin{theorem}
  \label{thm:correctnesstotal}
  Let $f : \mathbb{Z}^n \rightarrow \lbrace 0, 1\rbrace$ be any function. Suppose
  there is a ``post-processing'' function $g$ such that, for all $\sigma$ of
  length at most $n + 2m + 2$, $g(\mathcal{A}(\sigma)) = f(\freq \sigma)$. Then for
  all $x \in \mathbb{Z}^n$, $f(x) = g(\mathcal{A}(\kappa(\phi(x))))$.
\end{theorem}
\begin{proof}
  We proceed by induction on $i$, the largest non-negative integer such that $x_j
  < a_j$ for all $j > i$, and $x_i$.

  Suppose $i = 0$. Then $x = \phi(x)$ and the result follows immediately, as
  $\kappa(x)$ has length at most $n$. So suppose that this is not the case, and
  the result holds for all $x$ with smaller $i$ or the same $i$ and smaller
  $x_i$.

  Then by the construction of $o_i$ above, there exist $x$ and $y$ in
  $\prod_{j=1}^{i-1}\mathbb{Z}_{a_j}\times\lbrace0\rbrace^{n-i+1}$ and integers
  $r' < r$ such that $\mathcal{A}(\kappa(x + re_i)) = \mathcal{A}(\kappa(y +
  r'e_i))$, and $o_i = y - x$, $a_i = r - r'$.

  Now write $x = x' + x_ie_i + x''$, where $x'$ is zero on all indices at least
  $i$ and $x''$ is zero on all indices no greater than $i$. Then
  \begin{align*}
    \phi(x) &= \phi(x' + \lfloor x_i/a_i\rfloor o_i) + (x_i \bmod a_i)e_i + x''\\
    &= \phi(x' + o_i + (x_i - a_i)e_i + x'')
  \end{align*}
  and so by the inductive hypothesis: \[
    g(\mathcal{A}(\kappa ( \phi(x)))) = f(x' + o_i + (x_i - a_i)e_i + x'')
  \]
Now consider the following two streams:
  \begin{align*}
    \sigma_1 &= \kappa(x + re_i) \cdot \kappa(-x - re_i) \cdot \kappa(x)\\
    \sigma_2 &= \kappa(y + r'e_i) \cdot \kappa(-x - re_i) \cdot \kappa(x)
  \end{align*}
  Note that $x$ and $y$ are both supported on at most $m$ indices, so the length
  of these streams is at most $n + 2m + 2$ and so $g(\mathcal{A}(\sigma_1)) = f(\freq \sigma_1)$ and $g(\mathcal{A}(\sigma_2)) = f(\freq(\sigma_2))$. Furthermore, as $\mathcal{A}(\kappa(x + re_i)) = \mathcal{A}(\kappa(y + r'e_i))$, $\mathcal{A}(\sigma_1) = \mathcal{A}(\sigma_2)$, and so $f(\freq(\sigma_1) = f(\freq(\sigma_2))$.

  Now $\freq(\sigma_1) = x$, while 
  \begin{align*}
    \freq(\sigma_2) &= (y-x) - (r - r')e_i + x\\
    &= o_i - a_ie_i + x\\
    &= x' + o_i + (x_i - a_i)e_i + x''
  \end{align*}
  and so $f(x' + o_i + (x_i - a_i)e_i + x'') = f(x)$, and so \[
    g(\mathcal{A}(\kappa ( \phi(x)))) = f(x),
  \]
  completing the proof.
\end{proof}

\subsubsection{Turnstile-Sketching Equivalence}
\restate{thm:detreductiontotal}
\begin{proof}
  Let $\alg$ be the original algorithm. The algorithm will be to keep $\phi(x)$,
  where $x$ is the input vector (which by the previous sections we can do in $O(s
  + m\log n) \le O(s\log n)$ space), and then give $\alg$ $\kappa(\phi(x))$. By
  Theorem \ref{thm:correctnesstotal}, as $m \le s$, the output of $\alg$ will
  be $f(x)$.

  By the lemmas in Section~\ref{sec:totalreductionspace}, the conditions of
  Theorem~\ref{thm:modulespace} hold, and so this sketch can be stored in $s$
  space, and maintained in $O(S + s + m\log n) \le O(S + s\log n)$ space (as $m
  \le s$ by Corollary~\ref{cor:msizetotal}).  Recovering $f(x)$ from the sketch
  requires running $\alg$ on $\kappa(x)$, which takes $O(S + s)$ space.
\end{proof}

\subsection{Sketching General Stream Problems}
\subsubsection{Overview}
In order to prove an equivalence between linear sketches and turnstile
algorithms for general stream problems, we need to define parameters $a_i$ and
$o_i$ to instantiate the linear sketch $\phi \rightarrow \Z^n$.

Once we have defined these parameters we will prove the sketch is ``correct''
--- for every $x \in \Z^n$, there are streams with frequency $x, \phi(x)$ on
which $\alg$ outputs the same thing. We will then show that it is possible to
generate the parameters $a_i$ and $o_i$ in $O(s + m\log n)$ space, and
therefore by Theorem~\ref{thm:modulespace} we may maintain the sketch in this
space.

Finally, we will show that, using the streams described in the correctness
section, it is possible to recover a solution to any stream problem solved by
$\alg$ using the sketch.

\subsubsection{Defining the Parameters}
\label{sec:defineMpartial}
Along with the parameters $a_i$ and $o_i$, we also define ``prefix vectors''
$\pi_i$ for $i = 0, \dots, n$ and ``covering streams'' $\tau_i^x$ (for $x \in
\prod_{j = 1}^{i - 1}\Z_{a_j} \times \Z \times \lbrace 0 \rbrace^{n - i}$ and
$i = 1, \dots, n$) to be used in the recursive construction and in the later
proof of correctness.

These will be defined as the output of the following procedure, which proceeds
through the indices $i$ with backtracking.

Let $\pi_0$ be the empty stream. For $i = 1, \dots, n$:
\begin{itemize}
  \item We start by defining the covering streams $\tau_i^x$. Let $x_j$ be
    defined as the $j^\text{th}$ vector in  $x \in \prod_{j = 1}^{i -
    1}\Z_{a_j} \times \Z \times \lbrace 0 \rbrace^{n - i}$ in
    little-endian order. Then we define $\tau_i^{x_1} = \tau_i^0 = \pi_{i-1}$. For
    $j > 0$, we define $\tau_i^{x_j} = \tau_i^{x_{j-1}}\cdot\kappa(x_j - x_{j-1})$.

    Note that for any $j_1 < j_2$, as $\tau_i^{x_{j_1}}$ is a prefix of
    $\tau_i^{x_{j_2}}$ we may write $\tau_i^{x_{j_2}} =
    \tau_i^{x_{j_1}}\cdot\alpha$ for some stream $\alpha$ and $\freq \alpha$ will
    be equal to $x_{j_2} - x_{j_1}$.

  \item Let $j$ be the smallest integer such that there exists $j' < j$ such that
    $\alg(\tau_i^{x_j}) = \alg(\tau_i^{x_{j'}})$. Choose $a_i$, $o_i$ so that $x_j
    - x_{j'} = a_ie_i - o_i$. Note that $a_i \ge 0$ as $x_j$ is later than $x_{j'}$
    in little-endian order. If $a_i > 0$, set $\pi_i = \tau_i^{x_{j'}}$ and move on
    to the next $i$.

    If $a_i = 0$, let $i'$ be the largest index such that $(x_j - x_{j'})_{i'} >
    0$. Choose $a_{i'}$ and $o_{i'}$ so that $x_j - x_{j'} = a_{i'}e_{i'} -
    o_{i'}$, and set $\pi_{i'} = \tau_i^{x_{j'}}$, overwriting the old values of
    $a_{i'}$, $o_{i'}$, and $\pi_{i'}$. Then roll $i$ back to $i' + 1$ and continue
    from there.  
\end{itemize}

\begin{lemma}
  This procedure will terminate after a finite number of steps.
\end{lemma}
\begin{proof}
  After each iteration, either $i$ increases or $i$ is set to $i'+1$ with
  $a_{i'}$ reduced from its previous value. As the $a_i$ take values in the
  positive integers, the second can only happen finitely many times, and so the
  procedure will eventually terminate.
\end{proof}

\subsubsection{Space}
\label{sec:partialreductionspace}
\begin{lemma}
  $\prod_{i=1}^n a_i \le 2^s$.
\end{lemma}
\begin{proof}
  Consider the procedure from Section \ref{sec:defineMpartial}. In the final iteration
  (that is, when $a_n$ is defined rather than $i$ rolling back to some earlier
  index), $j$ was the smallest integer such that there existed $j'$ such that
  $\alg(\tau_n^{x_j}) = \alg(\tau_n^{x_j'})$, and $a_n = (x_j - x_{j'})_n$.

  As $j$ was the smallest integer such that this held, this implies that
  $\alg(\tau_n^{x_0}),\dots,\alg(\tau_n^{x_{j-1}})$ were all distinct states. As
  the sequence $x_k$ comes from iterating through the vectors in
  $\prod_{i=1}^{n-1}\Z_{a_i} \times \Z$ in little-endian order, $j$ is at least
  $\prod_{i=1}^{n-1}a_i \times (x_j)_n$. So as $a_n \le (x_j)_n$, there are at
  least $\prod_{i=1}^n a_i$ distinct states of $\alg$, and so the result follows.
\end{proof}
\noindent
Recall that $m = |\lbrace i \in \lbrack n : a_i > 1\rbrace|$.
\begin{corollary}
  \label{cor:msizepartial}
  $m \le s$
\end{corollary}
\begin{proof}
  This follows from the fact that the procedure that generates the $a_i$ will
  always roll back if it would generate an $a_i$ equal to 0, and therefore all
  the $a_i$ are positive integers.
\end{proof}

\begin{lemma}
  \label{lem:modulicalcpartial}
  We may calculate all the moduli $a_i$, while generating the stream $\pi_n$,
  in $O(S + s + m\log n)$ space.  
  \end{lemma}
\begin{proof}
  To execute the procedure that generates the $a_i$, we need to remember the
  values of all $a_j$ for $j < i$ (which we can store in $O(s + m\log n)$ space,
  as at most $m$ are greater than $1$ and their magnitudes sum to at most $2^s$),
  we need to remember $\alg(\pi_{i-1})$ (which takes $O(s)$ space) and then we
  need to find the pair $j > j'$ such that $\alg(\tau_i^{x_j}) =
  \alg(\tau_i^{x_{j'}})$. 

  Given the moduli $(a_j)_{j=1}^i$, we can generate the elements of the streams
  $\tau_i^{x_j}$ (from after $\pi_{i-1}$) on the fly in $O(S + s)$ space, as
  they just require marching through the elements of $\prod_{j=1}^{i-1}
  \Z_{a_i} \times \Z$ in little-endian order while executing the transition
  function of $\alg$ on each update, and the number of elements we go through
  is at most the number of distinct states of $\alg$.

  Therefore, we can find the pair in $O(S + s)$ space by running two copies of
  $\alg$ and feeding them the streams $\tau_i^{x_j}$ until we find a collision.
\end{proof}

\begin{lemma}
  \label{lem:overflowcalcpartial}
  For any $i$, $o_i$ can be calculated in $O(S + s + m\log n)$ space.
\end{lemma}
\begin{proof}
  First note that, as each $o_i$ is in
  $\prod_{j=1}^{i-1}\Z_{a_j}\times\Z^{n-i+1}$, they can be \emph{stored} in
  $O(s + m\log n)$ space by storing $(j,(o_i)_j)$ pairs as above.

  To calculate $o_i$, we may first calculate all the $a_i$ as above, and then run
  the procedure until the final time where it changes $a_i$. At that point we may
  read off $o_i$ (as we know $x_j$ and $x_{j'}$, as we tracked them while
  generating the streams $\tau_i^{x_j}$ and $\tau_i^{x_{j'}}$).
\end{proof}

\subsubsection{Correctness}
\begin{lemma}
\label{lem:negativeloop}
Let $\alpha$, $\beta$ be any pair of streams. Then there are infinitely many $l
\in \N$ such that $\alg(\alpha \cdot \beta^l) = \alg(\alpha \cdot
\beta^{2^s})$.
\end{lemma}
\begin{proof}
Consider the sequence of states $q_l = \alg(\alpha\cdot\beta^l)$. As there are
only $2^s$ distinct states, there is some state that recurs infinitely many
times, and that state must appear for some $l \le 2^s$. So let this $l =
2^s-k$. Each time this state appears, $\alg(\alpha\cdot\beta^{2^s})$ appears $k$
states later. So $\alg(\alpha \cdot \beta^{2^s})$ also appears infinitely many
times.
\end{proof}

\begin{theorem}
  \label{thm:approxcorrect}
  For all $x \in \Z^n$, there is a stream $\sigma$ such that
  $\freq\sigma = x$ and: \[
    \alg(\sigma) = \alg(\pi_n \cdot \overline{\pi_n}\cdot \kappa(o_1 -
    a_1e_1)^{2^s} \dots \kappa(o_n - a_ne_n)^{2^s} \cdot \kappa(\phi(x)))
  \]
\end{theorem}
\begin{proof}
For each $i \in \lbrack n\rbrack$, let $\psi_i$ be such that $\pi_i =
\pi_{i-1}\cdot \psi_i$ (recall that each $\pi_i$ is a prefix of the next), and
let $\rho_i$ be the stream found in the construction of $M$ such that
$\alg(\pi_i \cdot \rho_i) = \alg(\pi_i)$ and $\freq \rho_i = a_ie_i - o_i$. For
$y \in \N^n$, let $\xi_y$ be the following stream: \[
\psi_1 \cdot \rho_i^{y_1} \dots \psi_n \cdot \rho_n^{y_n}
\]
Then for all $y \in \N^n$, \[
\alg(\xi_y) = \alg(\psi_1 \dots \psi_n) = \alg(\pi_n)
\]
while $\freq \xi_y = \freq \pi_n + \sum_{i =1}^n y_i(a_ie_i - o_i)$. Next, for
$y \in \N^n$, let \[
\chi_{y} = \kappa(o_1 - a_1e_1)^{y_1} \dots \kappa(o_n - a_ne_n)^{y_n}
\]
so $\freq \chi_y = -\sum_{i =1}^n y_i(a_ie_i - o_i)$. We will prove
the theorem for a $\sigma$ of the form
\[
\sigma = \xi_y \cdot \overline{\pi_n} \cdot \chi_z \cdot \kappa(\phi(x))
\]
for carefully chosen $y$ and $z$.  Note that
\begin{align}
  \freq \sigma &= \freq \xi_y - \freq \pi_n + \freq \chi_z + \phi(x)\notag\\
  &= \phi(x) + \sum_{i =1}^n (y_i - z_i)(a_ie_i - o_i).\label{eq:freqsigma}
\end{align}
In particular, we will choose $z$ such that
\[ \alg(\pi_n \cdot \overline{\pi_n} \cdot \kappa(o_1 - a_1e_1)^{z_1}
  \dots \kappa(o_i - a_ie_i)^{z_i}) = \alg(\pi_n \cdot \overline{\pi_n}
  \cdot \kappa(o_1 - a_1e_1)^{2^s} \dots \kappa(o_i - a_ie_i)^{2^s})
\] for each $i \in \lbrack n\rbrack$.

We show that the theorem holds for such a $\sigma$ and $z$ by
induction on $i$, the largest non-negative integer such that
$0 \le x_j < a_j$ for all $j > i$.

Suppose $i = 0$. Then $x = \phi(x)$, so we can take $\sigma = \xi_{y} \cdot
\overline{\pi_n} \cdot \chi_z \cdot \kappa(\phi(x))$, where both $y$ and $z$
are the vectors with $2^s$ in every coordinate. Then
\begin{align*}
\alg(\sigma) &= \alg(\pi_n \cdot \overline{\pi_n} \cdot \chi_z \cdot \kappa(\phi(x))\\
&= \alg(\pi_n \cdot \overline{\pi_n} \cdot \kappa(o_1 - a_1e_1)^{2^s} \dots \kappa(o_n - a_ne_n)^{2^s} \cdot \phi(x))
\end{align*}
and by~\eqref{eq:freqsigma}, $\freq \sigma = \phi(x) = x$.  Finally,
the condition on $z$ is trivially satisfied, as $z_i = 2^s$ for each
$i$.

Now suppose $i > 0$, and the result holds for all $x$ with smaller $i$. Write $x = x' + x_ie_i + x''$, where $x'$ is zero on all indices at least
  $i$ and $x''$ is zero on all indices no greater than $i$. Then
  \begin{align*}
    \phi(x) &= \phi(x' + \lfloor x_i/a_i\rfloor o_i) + (x_i \bmod a_i)e_i + x''\\
    &= \phi(x' + \lfloor x_i/a_i\rfloor o_i + (x_i \bmod a_i)e_i + x'')
  \end{align*}
and by the inductive hypothesis there exists a $\sigma' = \xi_{y'} \cdot \overline{\pi_n} \cdot \chi_{z'} \cdot \kappa(\phi(x))$ such that \[
    \freq \sigma' = x' + \lfloor x_i/a_i\rfloor o_i + (x_i \bmod a_i)e_i + x''
\]
and \[
\alg(\sigma') = \alg(\pi_n \cdot \overline{\pi_n}\cdot \kappa(o_1 -
    a_1e_1)^{2^s} \dots \kappa(o_n - a_ne_n)^{2^s} \cdot \kappa(\phi(x)))
\]
with $z'$ such that \[
\alg(\pi_n \cdot \overline{\pi_n} \cdot \kappa(o_1 - a_1e_1)^{z_1'} \dots
\kappa(o_j - a_je_j)^{z_j'}) = \alg(\pi_n \cdot \overline{\pi_n} \cdot \kappa(o_1 -
a_1e_1)^{2^s} \dots \kappa(o_j - a_je_j)^{2^s})
\] for each $j \in \lbrack n\rbrack$.

Now, by Lemma~\ref{lem:negativeloop}, there are infinitely many $l \in \N$ such
that
\begin{align*}
&\alg(\pi_n \cdot \overline{\pi_n} \cdot \kappa(o_1 - a_1e_1)^{2^s} \dots
  \kappa(o_{i-1} - a_{i-1}e_{i-1})^{2^s}\cdot\kappa(o_i - a_ie_i)^l) \\
  =& \alg(\pi_n \cdot \overline{\pi_n} \cdot \kappa(o_1 -
a_1e_1)^{2^s} \dots \kappa(o_{i-1} - a_{i-1}e_{i-1})^{2^s}\cdot\kappa(o_i - a_ie_i)^{2^s})
  \end{align*}
so let $l$ be such that this holds and $l \ge z_i' -\lfloor x_i/a_i\rfloor$. We will
define $z$ to be $z'$ at every coordinate except that $z_i = l$. We will define
$y$ to be $y'$ except with $y_i = y_i' + l + \lfloor x_i/a_i\rfloor - z_i'$, so
$y$ is still in $\N^n$. 

Now let $\sigma = \xi_y \cdot \overline{\pi_n} \cdot \chi_z \cdot
\kappa(\phi(x))$. We will show this satisfies all the conditions required by
the inductive hypothesis. First, we show that $z$ obeys the desired property. For all $j \in \lbrack n\rbrack$, if $j < i$ it holds by the inductive hypothesis, as $z_j = z_j'$ for all $j < i$. Then, if $j = i$, 
\begin{align*}
\alg(\pi_n \cdot \overline{\pi_n} \cdot \kappa(o_1 - a_1e_1)^{z_1} \dots
\kappa(o_j - a_je_j)^{z_j}) &= \alg(\pi_n \cdot \overline{\pi_n} \cdot \kappa(o_1 - a_1e_1)^{2^s} \dots
\kappa(o_{i-1} - a_{i-1}e_{i-1})^{2^s}) \cdot \kappa(o_{i} - a_{i}e_i)^{l})\\
&= \alg(\pi_n \cdot \overline{\pi_n} \cdot \kappa(o_1 -
a_1e_1)^{2^s} \dots \kappa(o_i - a_ie_i)^{2^s})
\end{align*}
using the $j < i$ property and our choice of $l$.  For $j > i$, the
result again holds by the inductive hypothesis, as it holds for
$j = i$ and $z_j = z'_j$ for all $j > i$.

Now we show that $\alg(\sigma)$ takes the correct value.
\begin{align*}
\alg(\sigma) &= \alg(\xi_y \cdot \overline{\pi_n} \cdot \chi_z \cdot
\kappa(\phi(x)))\\
&= \alg(\pi_n \cdot \overline{\pi_n} \cdot \chi_z \cdot
\kappa(\phi(x)))\\
&= \alg(\pi_n \cdot \overline{\pi_n} \cdot \kappa(o_1 -
    a_1)^{2^s} \dots \kappa(o_n - a_n)^{2^s} \cdot \kappa(\phi(x)))
\end{align*}
by the property we just proved for $z$.

Finally we need to prove that $\freq \sigma = x$. The difference
between $\sigma$ and $\sigma'$ is that we replaced $\xi_{y'}$ with
$\xi_{y}$ and $\chi_{z'}$ with $\chi_z$, and $y, z$ each differ from
$y', z'$ only in coordinate $i$.  Therefore by~\eqref{eq:freqsigma},
\begin{align*}
  \freq \sigma &=  \freq \sigma' + (y_i - y_i' + z'_i - z_i)(a_ie_i - o_i)\\
               &= (x' + \lfloor x_i/a_i\rfloor o_i + (x_i \bmod a_i)e_i + x'') + (l + \lfloor x_i/a_i\rfloor - z_i + z'_i - l)(a_ie_i - o_i)\\
               &= x' + (\lfloor x_i/a_i\rfloor a_i + (x_i \bmod a_i))e_i + x''\\
&=  x' +  x_ie_i + x''\\
&= x
\end{align*}
completing the proof.
\end{proof}

\subsubsection{Sketching-Turnstile Equivalence}
\restate{thm:detreductionpartial}
\begin{proof}
  Let $\alg$ be the original algorithm. The new algorithm will be to
  construct $M$ and $\phi$ as above, and as we receive updates to the
  input vector $x$, maintain $\phi(x)$. By the Lemmas in
  Section~\ref{sec:partialreductionspace}, the conditions of
  Theorem~\ref{thm:modulespace} are satisfied, so this will require
  $O(S + s + m\log n) \le O(S + s \log n)$ space to compute (as
  $m \le s$ by Corollary~\ref{cor:msizepartial}).

  Then, at the end of the stream, we will input
  $\sigma^* := \pi_n \cdot \overline{\pi_n}\cdot \kappa(o_1 -
  a_1e_1)^{2^s} \dots \kappa(o_n - a_ne_n)^{2^s} \cdot \kappa(\phi(x))
  $ to $\alg$, and output whatever $\alg$ recovers from the resulting
  state (as we can compute $\pi_n$ we can also compute
  $\overline{\pi_n}$). By Theorem \ref{thm:approxcorrect}, there is a
  stream $\sigma$ with $\freq \sigma = x$ such that
  $\alg(\sigma^*) = \alg(\sigma)$, so as $\alg$ would have output a
  correct answer for $\sigma$ it will output the same correct answer
  when given $\sigma^*$.

  This recovery algorithm takes
  $O(S + s + m\log n) \le O(S + s\log n)$ space, as by
  Lemma~\ref{lem:modulicalcpartial} we can generate $\pi_n$ in that
  space (and therefore $\overline{\pi_n}$), even though we could not
  \emph{store} the whole stream. Similarly, we can generate the
  streams $\kappa(o_i - a_i)^{2^s}$ by generating $a_i$ and $o_i$ and
  using an $s$-bit counter to insert it the correct number of times.
  Other than computing the stream, we simply maintain $\alg$ under the
  stream $\sigma^*$ and apply the recovery algorithm, both of which
  use $S$ space by assumption.
\end{proof}

\bibliographystyle{alpha}
\bibliography{refs}

\end{document}